\RequirePackage{amsmath,amssymb}
\documentclass{llncs}

\usepackage{graphicx}
\usepackage{subcaption}

\usepackage{float}
\usepackage{tikz}
\usepackage{enumitem,linegoal}
\tikzset{
  circ/.style = {circle,draw,fill,inner sep=1pt},
  invisible/.style = {circle,draw=none,inner sep=0pt,font=\tiny}
}

\spnewtheorem{observation}{Observation}{\bfseries}{\itshape}
\spnewtheorem{Claim}{Claim}{\itshape}{\rmfamily}
\newenvironment{cproof}[1]{\par\indent{\textit{Proof.}}\space#1}{\hfill $\diamondsuit$}
\newcommand{\qe}{\tag*{$\square$}}

\title{On Contact Graphs of Paths on a Grid}
\author{Zakir Deniz\inst{1}
\and Esther Galby\inst{2} \and Andrea Munaro\inst{2} \and Bernard Ries\inst{2}}

\institute{Duzce University, Department of Mathematics, Duzce, Turkey, \texttt{zakirdeniz@duzce.edu.tr} 
\and University of Fribourg, Department of Informatics, Decision Support $\&$ Operations Research,
Fribourg, Switzerland,\\ \texttt{esther.galby@unifr.ch, andrea.munaro@unifr.ch, bernard.ries@unifr.ch}}

\begin{document}
\maketitle
\setcounter{footnote}{0}
\parindent=0cm

\begin{abstract}
In this paper we consider \textit{Contact graphs of Paths on a Grid} (\textit{CPG graphs}), i.e. graphs for which there exists a family of interiorly disjoint paths on a grid in one-to-one correspondence with their vertex set such that two vertices are adjacent if and only if the corresponding paths touch at a grid-point. Our class generalizes the well studied class of VCPG graphs (see \cite{felsner1}). We examine CPG graphs from a structural point of view which leads to constant upper bounds on the clique number and the chromatic number. Moreover, we investigate the recognition and 3-colorability problems for $B_0$-CPG, a subclass of CPG. We further show that CPG graphs are not necessarily planar and not all planar graphs are CPG. 
\end{abstract}


\section{Introduction}

Asinowski et al. \cite{asinowski} introduced the class of \textit{vertex intersection graphs of paths on a grid}, referred to as \textit{VPG graphs}. An undirected graph $G=(V,E)$ is called a \textit{VPG graph} if one can associate a path on a grid with each vertex such that two vertices are adjacent if and only if the corresponding paths intersect on at least one grid-point. It is not difficult to see that the class of VPG graphs coincides with the class of string graphs, i.e. intersection graphs of curves in the plane~(see~\cite{asinowski}).

A natural restriction which was forthwith considered consists in limiting the number of \textit{bends} (i.e. $90$ degrees turns at a grid-point) that the paths may have: an undirected graph $G=(V,E)$ is a \textit{$B_k$-VPG graph}, for some integer $k\geq 0$, if one can associate a path on a grid having at most $k$ \textit{bends} with each vertex such that two vertices are adjacent if and only if the corresponding paths intersect on at least one grid-point. Since their introduction, $B_k$-VPG have been extensively studied (see \cite{alcon,asinowski,chaplick,chaplick12,cohen1,cohen2,Felsner,francis,Golumbic1,gonca,heldt1}).

A notion closely related to intersection graphs is that of \textit{contact graphs}. Such graphs can be seen as a special type of intersection graphs of geometrical objects in which these objects are not allowed to have common interior points but only to touch each other. Contact graphs of various types of objects have been studied in the literature (see, e.g., \cite{felsner1,castro,Fraysseix1,Hlineny,Hlineny1bis,Hlineny1}). In this paper, we consider \textit{Contact graphs of Paths on a Grid} (\textit{CPG graphs} for short) which are defined as follows. A graph $G$ is a \textit{CPG graph} if the vertices of $G$ can be represented by a family of interiorly disjoint paths on a grid, two vertices being adjacent in $G$ if and only if the corresponding paths touch, i.e. share a grid-point which is an endpoint of at least one of the two paths (see Fig.~\ref{contact}). Note that this class is hereditary, i.e. closed under vertex deletion. Similarly to VPG, a $B_k$-CPG graph is a CPG graph admitting a representation in which each path has at most $k$ bends. Clearly, any $B_k$-CPG graph is also a $B_k$-VPG graph.

\begin{figure}[htb]
\centering
\begin{subfigure}{.5\textwidth}
\centering
\begin{tikzpicture}
\draw[thick] (0,-.5)--(0,.5);
\draw[dotted,thick] (0,-1)--(0,-.5);
\draw[dotted,thick] (0,1)--(0,.5);
\draw[->,thick,>=stealth] (-.75,0)--(0,0);
\draw[dotted,thick] (-1.25,0)--(-.75,0);

\draw[dotted,thick] (.5,0) -- (1,0);
\draw[->,thick,>=stealth] (1,0) -- (1.75,0);
\draw[<-,thick,>=stealth] (1.75,0) -- (2.5,0);
\draw[dotted,thick] (2.5,0) -- (3,0);
\end{tikzpicture}
\caption{Allowed contacts.}
\label{contact1}
\end{subfigure}
\hspace*{1cm}
\begin{subfigure}{.4\textwidth}
\centering
\begin{tikzpicture}
\draw[thick] (0,0) -- (.5,0) -- (.5,.5);
\draw[dotted,thick] (-.25,0) -- (0,0);
\draw[dotted,thick] (.5,.5) -- (.5,.75);
\draw[thick,lightgray] (1,0) -- (.5,0) -- (.5,-.5);
\draw[dotted,thick,lightgray] (1.25,0) -- (1,0);
\draw[dotted,thick,lightgray] (.5,-.5) -- (.5,-.75);
\node[invisible] at (0,-1.25) {};
\end{tikzpicture}
\caption{Forbidden contact.}
\label{noknockknee}
\end{subfigure}
\caption{Examples of types of contact between two paths (the endpoints of a path are marked by an arrow).}
\label{contact}
\end{figure}
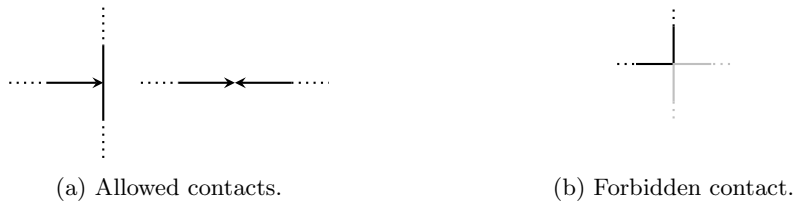 
Aerts and Felsner \cite{felsner1} considered a similar family of graphs, namely those admitting a \textit{Vertex Contact representation of Paths on a Grid} (\textit{VCPG} for short). The vertices of such graphs can be represented by a family of interiorly disjoint paths on a grid, but the adjacencies are defined slightly differently: two vertices are adjacent if and only if the endpoint of one of the corresponding paths touches an interior point of the other corresponding path (observe that this is equivalent to adding the constraint forbidding two paths from having a common endpoint, i.e. contacts as in Fig.~\ref{contact1} on the right). This class has been considered by other authors as well (see \cite{chaplick13,chaplick12,Felsner,gonca,kabourov}). 

It is not difficult to see that graphs admitting a VCPG are planar (see \cite{felsner1}) and it immediately follows from the definition that those graphs are CPG graphs. This containment is in fact strict even when restricted to planar CPG graphs, as there exist, in addition to nonplanar CPG graphs, planar graphs which are CPG but do not admit a VCPG. 

To the best of our knowledge, the class of CPG graphs has never been studied in itself and our present intention is to provide some structural properties (see Section \ref{sec:structure}). By considering a specific weight function on the vertices, we provide upper bounds on the number of edges in CPG graphs as well as on the clique number and the chromatic number (see Section \ref{sec:structure}). In particular, we show that $B_0$-CPG graphs are 4-colorable and that {\sc 3-colorability} restricted to $B_0$-CPG is $\mathsf{NP}$-complete (see Section \ref{sec:coloring}). We further prove that recognizing $B_0$-CPG graphs is $\mathsf{NP}$-complete. Additionally, we show that the classes of CPG graphs and planar graphs are incomparable (see Section~\ref{sec:planar}). 


\section{Preliminaries}
\label{sec:prelim}

Throughout this paper, all considered graphs are undirected, finite and simple. For any graph theoretical notion not defined here, we refer the reader to \cite{diestel}.

Let $G=(V,E)$ be a graph with vertex set $V$ and edge set $E$. The \textit{degree} of a vertex $v\in V$, denoted by $d(v)$, is the number of neighbors of $v$ in $G$. A graph $G$ is \textit{$k$-regular} if the degree of every vertex in $G$ is $k\geq 0$. A \textit{clique} (resp. \textit{stable set}) in $G$ is a set of pairwise adjacent (resp. nonadjacent) vertices. The graph obtained from $G$ by deleting a vertex $v\in V$ is denoted by $G - v$. For a given graph $H$, $G$ is \textit{$H$-free} if it contains no induced subgraph isomorphic to $H$. 

As usual, $K_n$ (resp. $C_n$) denotes the complete graph (resp. chordless cycle) on $n$ vertices and $K_{m,n}$ denotes the complete bipartite graph with bipartition $(V_1,V_2)$ such that $|V_1| = m$ and $|V_2| = n$. Given a graph $G$, the \textit{line graph of $G$}, denoted by $L(G)$, is the graph such that each vertex $v_e$ in $L(G)$ corresponds to an edge $e$ in $G$ and two vertices are adjacent in $L(G)$ if and only if their corresponding edges in $G$ have a common endvertex.

A graph $G$ is \textit{planar} if it can be drawn in the plane without crossing edges; such a drawing is then called a \textit{planar embedding} of $G$. A planar embedding divides the plane into several regions referred to as \textit{faces}. A planar graph is \textit{maximally planar} if adding any edge renders it nonplanar. A maximally planar graph has exactly $2n-4$ faces, where $n$ is the number of vertices in the graph. A graph $H$ is a \textit{minor} of a graph $G$, if $H$ can be obtained from $G$ by deleting edges and vertices and by contracting edges. It is well-known that a graph is planar if and only if it does not contain $K_5$ or $K_{3,3}$ as a minor \cite{diestel}.

A \textit{coloring} of a graph $G$ is a mapping $\mathbf{c}$ associating with every vertex $u$ an integer $\mathbf{c}(u)$, called a \textit{color}, such that $\mathbf{c}(v)\neq \mathbf{c}(u)$ for every edge $uv$. If at most $k$ distinct colors are used, $\mathbf{c}$ is called a \textit{$k$-coloring}. The smallest integer $k$ such that $G$ admits a $k$-coloring is called the \textit{chromatic number} of $G$, denoted by $\chi(G)$.

Consider a rectangular grid $\mathcal{G}$ where the horizontal lines are referred to as \textit{rows} and the vertical lines as \textit{columns}. The grid-point lying on row $x$ and column $y$ is denoted by $(x,y)$. An \textit{interior point} of a path $P$ on $\mathcal{G}$ is a point belonging to $P$ and different from its endpoints; the \textit{interior} of $P$ is the set of all its interior points. A graph $G=(V,E)$ is \textit{CPG} if there exists a collection $\mathcal{P}$ of interiorly disjoint paths on a grid $\mathcal{G}$ such that $\mathcal{P}$ is in one-to-one correspondence with $V$ and two vertices are adjacent in $G$ if and only if the corresponding paths touch; if every path in $\mathcal{P}$ has at most $k$ bends, $G$ is $B_k$-CPG. The pair $\mathcal{R} = (\mathcal{G},\mathcal{P})$ is a \textit{CPG representation} of $G$, and more specifically a \textit{$k$-bend CPG representation} if every path in $\mathcal{P}$ has at most $k$ bends. In the following, the path representing some vertex $u$ in a CPG representation $\mathcal{R}$ of a graph $G$ is denoted by $P_u^{\mathcal{R}}$, or simply $P_u$ if it is clear from the context.

Let $G=(V,E)$ be a CPG graph and $\mathcal{R} = (\mathcal{G},\mathcal{P})$ be a CPG representation of $G$. A grid-point $p$ is of \textit{type I} if it corresponds to an endpoint of four paths in $\mathcal{P}$ (see Fig. \ref{intpointI}), and of \textit{type II} if it corresponds to an endpoint of two paths in $\mathcal{P}$ and an interior point of a third path in $\mathcal{P}$ (see Fig. \ref{intpointII}).

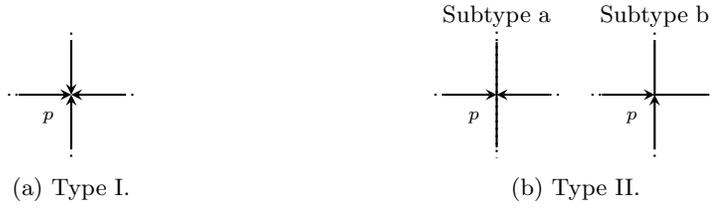
\begin{figure}[ht]
\centering  
\begin{subfigure}[b]{.45\textwidth}
\centering
\begin{tikzpicture}[scale=.7]
\node at (0,0) (p)  [label=below left:{\scriptsize $p$}]  {};
\draw[thick,<-,>=stealth] (0,0)--(0,1);
\draw[thick,dotted] (0,1)--(0,1.2);
\draw[thick,<-,>=stealth] (0,0)--(1,0);
\draw[thick,dotted] (1,0)--(1.2,0);
\draw[thick,<-,>=stealth] (0,0)--(0,-1);
\draw[thick,dotted] (0,-1)--(0,-1.2);
\draw[thick,<-,>=stealth] (0,0)--(-1,0);
\draw[thick,dotted] (-1.2,0)--(-1,0);
\end{tikzpicture}
\caption{Type I.}
\label{intpointI}
\end{subfigure}
\hspace*{1cm}
\begin{subfigure}[b]{.45\textwidth}
\centering
\begin{tikzpicture}[scale=.7]
\node at (0,0) (p)  [label=below left:{\scriptsize $p$}]  {};
\node[draw=none] at (0,1.5) {Subtype a};
\draw[thick,-,>=stealth] (0,1)--(0,-1);
\draw[thick,dotted] (0,1.2)--(0,-1.2);
\draw[thick,<-,>=stealth] (0,0)--(1,0);
\draw[thick,dotted] (1,0)--(1.2,0);
\draw[thick,<-,>=stealth] (0,0)--(-1,0);
\draw[thick,dotted] (-1,0)--(-1.2,0);

\node at (3,0) (p)  [label=below left:{\scriptsize $p$}]  {};
\node[draw=none] at (3,1.5) {Subtype b};
\draw[thick,-,>=stealth] (3,1)--(3,0)--(4,0);
\draw[thick,dotted] (3,1)--(3,1.2);
\draw[thick,dotted] (4,0)--(4.2,0);
\draw[thick,<-,>=stealth] (3,0)--(3,-1);
\draw[thick,dotted] (3,-1)--(3,-1.2);
\draw[thick,<-,>=stealth] (3,0)--(2,0);
\draw[thick,dotted] (1.8,0)--(2,0);
\end{tikzpicture}
\caption{Type II.}
\label{intpointII}
\end{subfigure}
\caption{Two types of grid-points.}
\label{intpoint}
\end{figure}

For any grid-point $p$, we denote by $\tau (p)$ the number of edges in the subgraph induced by the vertices whose corresponding paths contain or have $p$ as an endpoint. Note that this subgraph is a clique and so $\tau(p)=\binom{j}{2}$ if $j$ paths touch at grid-point $p$.

For any path $P$, we denote by $\mathring{P}$ (resp. $\partial(P)$) the interior (resp. endpoints) of $P$. For a vertex $u \in V$, we define the \textit{weight of $u$ with respect to} $\mathcal{R}$, denoted by $w_u^{\mathcal{R}}$ or simply $w_u$ if it is clear from the context, as follows. Let $q_u^i$ ($i=1,2$) be the endpoints of the corresponding path $P_u$ in $\mathcal{P}$ and consider, for $i=1,2$, 
\[w_u^i = |\{P \in \mathcal{P} ~|~ q_u^i \in \mathring{P}\}| + \frac{1}{2} \cdot |\{P \in \mathcal{P} ~|~ P \neq P_u \text{ and } q_u^i \in \partial(P)\}|.\]
Then $w_u = w_u^1 + w_u^2$.

\begin{observation}
\label{wui}
Let $G=(V,E)$ be a CPG graph and $\mathcal{R} = (\mathcal{G},\mathcal{P})$ be a CPG representation of $G$. For any vertex $u \in V$ and $i=1,2$, $w_u^i \leq \frac{3}{2}$ where equality holds if and only if $q_u^i$ is a grid-point of type I or II.
\end{observation} 

Indeed, the contribution of $q_u^i$ to $w_u^i$ is maximal if all four grid-edges containing $q_u^i$ are used by paths of $\mathcal{P}$, which may only happen when $q_u^i$ is a grid-point of type I or II. \\

\textit{Remark.} In fact, we have $w_u^i \in \{0,\frac{1}{2},1,\frac{3}{2}\}$ for any vertex $u \in V$ and $i=1,2$.

\begin{observation}
\label{lowerbound}
Let $G=(V,E)$ be a CPG graph and $\mathcal{R} = (\mathcal{G},\mathcal{P})$ be a CPG representation of $G$. Then \[|E| \leq \sum\limits_{u \in V} w_u,\] where equality holds if and only if all paths of $\mathcal{P}$ pairwise touch at most once. 
\end{observation}

Indeed, if $uv \in E$, we may assume that either an endpoint of $P_u$ touches the interior of $P_v$, or $P_u$ and $P_v$ have a common endpoint. In the first case, the edge $uv$ is fully accounted for in the weight of $u$, and in the second case, the edge $uv$ is accounted for in both $w_u$ and $w_v$ by one half. The characterization of equality then easily follows.


\section{Structural Properties of CPG Graphs}
\label{sec:structure}

In this section, we investigate CPG graphs from a structural point of view and present some useful properties which we will further exploit.

\begin{lemma}
\label{degree}
A CPG graph is either 6-regular or has a vertex of degree at most~5.
\end{lemma}

\begin{proof}
If $G=(V,E)$ is a CPG graph and $\mathcal{R}$ is a CPG representation of $G$, by combining Observations \ref{wui} and \ref{lowerbound}, we obtain 
\[\sum\limits_{u \in V} d(u) = 2|E| \leq 2\sum\limits_{u \in V} w_u \leq 2\sum\limits_{u \in V} \bigg( \frac{3}{2} + \frac{3}{2}\bigg) = 6|V|. \qe \]
\end{proof}

\textit{Remark.} We can show that there exists an infinite family of 6-regular CPG graphs. Due to lack of space, this proof is here omitted but can be found in Section \ref{app:6reg} of the Appendix.\\

For $B_1$-CPG graphs, we can strengthen Lemma \ref{degree} as follows.

\begin{proposition}
\label{deg5}
Every $B_1$-CPG graph has a vertex of degree at most 5.
\end{proposition}

\begin{proof}
Let $G=(V,E)$ be a $B_1$-CPG graph and $\mathcal{R}$ be a $1$-bend CPG representation of $G$. Denote by $p$ the upper-most endpoint of a path among the left-most endpoints in $\mathcal{R}$, and by $P_x$ (with $x \in V$) an arbitrary path having $p$ as an endpoint. Since $\mathcal{R}$ is a $1$-bend CPG representation, no path uses the grid-edge on the left of $p$, for otherwise $p$ would not be a left-most endpoint. Therefore, $p$ contributes to the weight of $x$ with respect to $\mathcal{R}$ by at most $1$ and, by Observations \ref{wui} and \ref{lowerbound}, we have
\[\sum\limits_{u \in V} d(u) = 2|E| \leq 2 (w_x + \sum\limits_{u \neq x} w_u) \leq 6|V| - 1,\]
which implies the existence of a vertex of degree at most 5. \qed
\end{proof}

A natural question that arises when considering CPG graphs is whether they may contain large cliques. It immediately follows from Observation \ref{lowerbound} that CPG graphs cannot contain $K_n$, for $n\geq 8$. This can be further improved as shown in the next result.

\begin{theorem}
\label{K7}
CPG graphs are $K_7$-free.
\end{theorem}

\begin{proof}
Since the class of CPG graphs is hereditary, it is sufficient to show that $K_7$ is not a CPG graph. Suppose, to the contrary, that $K_7$ is a CPG graph and consider a CPG representation $\mathcal{R} = (\mathcal{G},\mathcal{P})$ of $K_7$. Observe first that the weight of every vertex with respect to $\mathcal{R}$ must be exactly $2 \cdot 3/2$, as otherwise by Observation \ref{wui}, we would have $\sum_{u\in V} w_u < 3|V|=21=|E|$ which contradicts Observation~\ref{lowerbound}. This implies in particular that every grid-point corresponding to an endpoint of a path is either of type I or II. Furthermore, any two paths must touch at most once, for otherwise by Observation~\ref{lowerbound}, $|E| < \sum_{u\in V} w_u = 3|V|=|E|$. Hence, if we denote by $P_I$ (resp. $P_{II}$) the set of grid-points of type I (resp. type II), then since $\tau (p) = 6$ for all $p \in P_I$ and $\tau (p) = 3$ for all $p \in P_{II}$, we have that $6|P_I| + 3|P_{II}| = 21$, which implies $|P_{II}| \neq 0$. Suppose that there exists a path $P_u$ having one endpoint corresponding to a grid-point of type I and the other corresponding to a grid-point of type II. Since the corresponding vertex $u$ has degree 6, $P_u$ must then properly contain an endpoint of another path which, as first observed, necessarily corresponds to a grid-point of type II. But vertex $u$ would then have degree $3+2+2$ as no two paths touch more than once, a contradiction. Hence, every path has both its endpoints of the same type. But then, $|P_I| = 0$; indeed, if there exists a path having both its endpoints of type I, since no two paths touch more than once, this implies that every path has both its endpoints of type I, i.e. $|P_{II}| =0$, a contradiction. Now, if we consider each grid-point of type II as a vertex and connect any two such vertices when the corresponding grid-points belong to a same path, then we obtain a planar embedding of a 4-regular graph on 7 vertices. But this contradicts the fact that every 4-regular graph on 7 vertices contains $K_{3,3}$ as a minor (a proof of this result can be found in Section \ref{4-regular} of the Appendix). \qed
\end{proof}

However, CPG graphs may contain cliques on 6 vertices as shown in Proposition~\ref{K6}. Due to lack of space, its proof is omitted here and can be found in Section \ref{app:k6} of the Appendix.

\begin{proposition}
\label{K6}
$K_6$ is in $B_2$-CPG $\backslash B_1$-CPG.
\end{proposition}

We conclude this section with a complexity result pointing towards the fact that there may not be a polynomial characterization of $B_0$-CPG graphs. Let us first introduce rectilinear planar graphs: a graph $G$ is \textit{rectilinear planar} if it admits a rectilinear planar drawing, i.e. a drawing mapping each edge to a horizontal or vertical segment. 

\begin{theorem}
\label{thm:rec}
{\sc Recognition} is $\mathsf{NP}$-complete for $B_0$-CPG graphs.
\end{theorem}

\begin{proof}
We show that a graph $G$ is rectilinear planar if and only if its line graph $L(G)$ is $B_0$-CPG. As {\sc Recognition} for rectilinear planar graphs was shown to be $\mathsf{NP}$-complete in \cite{garg}, this concludes the proof. Suppose $G$ is a rectilinear planar graph and let $\mathcal{D}$ be the collection of horizontal and vertical segments in a rectilinear planar drawing of $G$. It is not difficult to see that the contact graph of $\mathcal{D}$ is isomorphic to $L(G)$. Conversely, assume that $L(G)$ is a $B_0$-CPG graph and consider a 0-bend CPG representation $\mathcal{R} = (\mathcal{G},\mathcal{P})$ of $L(G)$. Since $L(G)$ is $K_{1,3}$-free \cite{beineke}, every path in $\mathcal{P}$ has at most two contact points. Thus, by eventually shortening paths, we may assume that contacts only happen at endpoints of paths. Therefore, $\mathcal{R}$ induces a rectilinear planar drawing of $G$, where each vertex corresponds to a contact point in $\mathcal{R}$ and each edge is mapped to its corresponding path in $\mathcal{P}$. \qed
\end{proof}


\section{Planar CPG Graphs}
\label{sec:planar}

In this section, we focus on planar graphs and their relation with CPG graphs. In particular, we show that not every planar graph is CPG and not all CPG graphs are planar.\footnote{We can further show that not all CPG graphs are 1-planar as $K_7 - E(K_3)$ is CPG but not 1-planar \cite{korzhik}.}

\begin{lemma}
\label{trianglefree}
If $G$ is a CPG graph for which there exists a CPG representation containing no grid-point of type I or II.a, then $G$ is planar. In particular, if $G$ is a triangle-free CPG graph, then $G$ is planar.
\end{lemma}

\begin{proof}
Let $G=(V,E)$ be a CPG graph for which there exists a CPG representation $\mathcal{R}$ containing no grid-point of type I or II.a. By considering each path of $\mathcal{R}$ as a curve in the plane, it follows that $G$ is a curve contact graph having a representation (namely $\mathcal{R}$) in which any point in the plane belongs to at most three curves. Furthermore, whenever a point in the plane belongs to the interior of a curve $\mathcal{C}$ and corresponds to an endpoint of two other curves, then those two curves lie on the same side of $\mathcal{C}$ (recall that there is no grid-point of type II.a). Hence, it follows from Proposition 2.1 in \cite{Hlineny} that $G$ is planar. 

If $G$ is a triangle-free CPG graph, then no CPG representation of $G$ contains grid-points of type I or II.a. Hence, $G$ is planar. \qed
\end{proof}

\textit{Remark.} Since $K_{3,3}$ is a triangle-free nonplanar graph, it follows from Lemma \ref{trianglefree} that $K_{3,3}$ is not CPG. Therefore, CPG graphs are $K_{3,3}$-free. Observe however that for any $k \geq 0$, $B_k$-CPG is not a subclass of planar graphs as there exist $B_0$-CPG graphs which are not planar (see Fig. \ref{noplanar}).

\begin{figure}[htb]
\centering
\begin{subfigure}[b]{.45\textwidth}
\centering   
\begin{tikzpicture}[scale=.8]
\node[circ,label=left:{\tiny $1$}] at (0,0) (a) {};
\node[circ,label=left:{\tiny $2$}] at (-1,0) (b) {};
\node[circ,label=left:{\tiny $3$}] at (1,0) (c) {};
\node[circ,label=above:{\tiny $4$}] at (-.5,.9) (d) {};
\node[circ,label=above:{\tiny $5$}] at (.5,.9) (e) {};
\node[circ,label=below:{\tiny $6$}] at (-.5,-.9) (f) {};
\node[circ,label=below:{\tiny $7$}] at (.5,-.9) (g) {};

\draw[-] (a) -- (d)
(a) -- (e)
(a) -- (f)
(a) -- (g)
(b) -- (d) 
(b) -- (e)
(b) -- (f)
(c) -- (d) 
(c) -- (e)
(c) -- (g)
(f) -- (g) node[below,midway] {\tiny $e$};
\end{tikzpicture}
\caption{A nonplanar graph $G$.}
\end{subfigure}
\hspace*{1cm}
\begin{subfigure}[b]{.45\textwidth}
\centering
\begin{tikzpicture}[scale=.8]
\draw[thick,<->,>=stealth] (0,0)--(0,2) node[midway,left] {\tiny $P_5$};
\draw[thick,<->,>=stealth] (2,0)--(2,2) node[midway,right] {\tiny $P_4$};
\draw[thick,<->,>=stealth] (0,0)--(2,0) node[near end,above] {\tiny $P_3$};
\draw[thick,<->,>=stealth] (0,1)--(2,1) node[near end,above] {\tiny $P_1$};
\draw[thick,<->,>=stealth] (0,2)--(2,2) node[near end,above] {\tiny $P_2$};
\draw[thick,<->,>=stealth] (1,0)--(1,1) node[midway,left] {\tiny $P_7$};
\draw[thick,<->,>=stealth] (1,1)--(1,2) node[midway,left] {\tiny $P_6$};
\node[invisible] at (0,-.3) {};
\end{tikzpicture}
\caption{A $0$-bend CPG representation of $G$.}
\end{subfigure} 
\caption{A $B_0$-CPG graph containing $K_{3,3}$ as a minor (contract the edge $e$).}
\label{noplanar}
\end{figure}

It immediatly follows from \cite{chaplick12} that all triangle-free planar graphs are $B_1$-CPG; hence, we have the following corollary.

\begin{corollary}
If a graph $G$ is triangle-free, then $G$ is planar if and only if $G$ is $B_1$-CPG.
\end{corollary}

The next result allows us to detect planar graphs that are not CPG.

\begin{lemma}
\label{maxdeg3}
Let $G = (V,E)$ be a planar graph. If $G$ is a CPG graph, then $G$ has at most $4|V| - 2f + 4$ vertices of degree at most 3, where $f$ denotes the number of faces of $G$. In particular, if $G$ is maximally planar, then $G$ has at most 12 vertices of degree at most 3.
\end{lemma}

\begin{proof}
Let $G = (V,E)$ be a planar CPG graph and $\mathcal{R}=(\mathcal{G},\mathcal{P})$ a CPG representation of $G$. Denote by $U$ the subset of vertices in $G$ of degree at most 3. If a path $P_u$, with $u \in U$, touches every other path in $\mathcal{P}$ at most once, then, since at least one endpoint of $P_u$ is then not a grid-point of type I or II, the weight of $u$ with respect to $\mathcal{R}$ is at most $3/2 + 1$. Thus, if we assume that this is the case for all paths whose corresponding vertex is in $U$, we have by Observation \ref{lowerbound}
\[|E| \leq \bigg(\frac{3}{2} + 1\bigg) |U| + 3(|V| - |U|) = 3|V| - \frac{|U|}{2}.\]
On the other hand, if there exists $u \in U$ such that $P_u$ touches some path more than once, then the above inequality still holds as the corresponding edge is already accounted for. Using the fact that $f = |E| - |V| + 2$ (Euler's formula), we obtain the desired upper bound. Moreover, if $G$ is maximally planar, then $f = 2|V| - 4$ and so $|U| \leq 12$. \qed
\end{proof}

\textit{Remark.} In Fig. \ref{notCPG}, we give an example of a maximally planar graph which is not CPG due to Lemma \ref{maxdeg3}. It is constructed by iteratively adding a vertex in a triangular face, starting from the triangle, so that it has exactly 13 vertices of degree 3. There exist however maximally planar graphs which are CPG (see Fig. \ref{pCPG}). Note that maximally planar graphs do not admit a VCPG \cite{felsner1}.

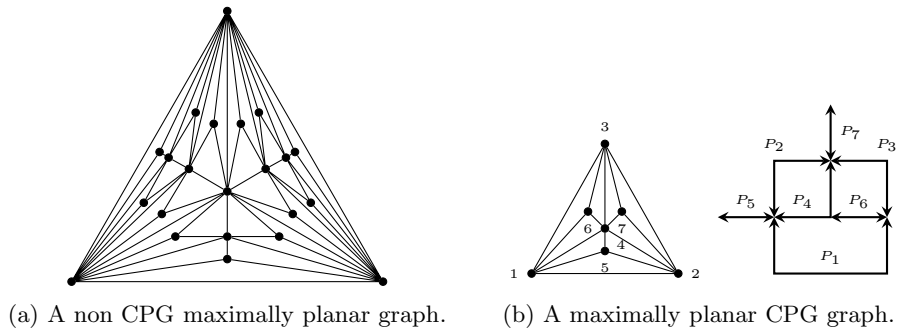
\begin{figure}[htb]
\centering
\begin{subfigure}[b]{.48\textwidth}
\centering
\begin{tikzpicture}[scale=0.3]
\node[circ] (a) at (0,0) {};
\node[circ] (d) at (0,-2) {};
\node[circ] (g) at (0,-3) {};
\node[circ] (b) at (1.7,1) {};
\node[circ] (c) at (-1.7,1) {};
\node[circ] (e) at (2.6,1.5) {};
\node[circ] (f) at (-2.6,1.5) {};
\node[circ] (h) at (3,1.75) {};
\node[circ] (i) at (-3,1.75) {};
\node[circ] (j) at (6.9,-4) {};
\node[circ] (l) at (-6.9,-4) {};
\node[circ] (k) at (0,8) {};
\node[circ] (1) at (2.3,-2) {};
\node[circ] (2) at (-2.3,-2) {};
\node[circ] (3) at (-2.9,-1) {};
\node[circ] (4) at (-0.6,3) {};
\node[circ] (5) at (0.6,3) {};
\node[circ] (6) at (2.9,-1) {};
\node[circ] (7) at (-3.7,-0.5) {};
\node[circ] (8) at (-1.4,3.5) {};
\node[circ] (9) at (1.4,3.5) {};
\node[circ] (10) at (3.7,-0.5) {};

\draw[-] (a) -- (b)
(a) -- (c)
(a) -- (d)
(a) -- (j) 
(a) -- (k)
(a) -- (l)
(a) -- (1)
(a) -- (2) 
(a) -- (3)
(a) -- (4)
(a) -- (5)
(a) -- (6)
(b) -- (e)
(b) -- (j)
(b) -- (k)
(b) -- (5)
(b) -- (6)
(b) -- (9)
(b) -- (10)
(c) -- (f)
(c) -- (k)
(c) -- (l)
(c) -- (3)
(c) -- (4)
(c) -- (7) 
(c) -- (8)
(d) -- (g)
(d) -- (j)
(d) -- (l)
(d) -- (1) 
(d) -- (2)
(e) -- (h) 
(e) -- (j)
(e) -- (k)
(e) -- (9)
(e) -- (10)
(f) -- (i)
(f) -- (k)
(f) -- (l)
(f) -- (7) 
(f) -- (8)
(g) -- (j)
(g) -- (l)
(h) -- (j)
(h) -- (k)
(i) -- (k)
(i) -- (l)
(j) -- (k)
(j) -- (l)
(j) -- (1)
(j) -- (6)
(j) -- (10)
(k) -- (l)
(k) -- (4)
(k) -- (5)
(k) -- (8)
(k) -- (9)
(l) -- (2)
(l) -- (3)
(l) -- (7);

\end{tikzpicture}
\caption{A non CPG maximally planar graph.}
\label{notCPG}
\end{subfigure}
\hspace*{.5cm}
\begin{subfigure}[b]{.43\textwidth}
\centering
\begin{tikzpicture}[scale=.75]
\node[circ,label=left:{\tiny 1}] (a) at (-1.3,0) {};
\node[circ,label=right:{\tiny 2}] (b) at (1.3,0) {};
\node[circ,label=above:{\tiny 3}] (c) at (0,2.3) {};
\node[circ,label=below right:{\tiny 4}] (d) at (0,.8) {};
\node[circ,label=below:{\tiny 5}] (e) at (0,.4) {};
\node[circ,label=below:{\tiny 6}] (f) at (-.3,1.1) {};
\node[circ,label=below:{\tiny 7}] (g) at (.3,1.1) {};

\draw[-] (a) -- (b)
(a)-- (c)
(a) -- (d)
(a) -- (e)
(a) -- (f)
(b) -- (c) 
(b) -- (d)
(b) -- (e) 
(b) -- (g)
(c) -- (d)
(c) -- (f)
(c) -- (g)
(d) -- (e)
(d) -- (f)
(d) -- (g);

\draw[<->,thick,>=stealth] (3,1) -- (3,0) -- node[midway,above] {\tiny $P_1$} (5,0) -- (5,1);
\draw[<->,thick,>=stealth] (3,1) -- (3,2) node[above] {\tiny $P_2$} -- (4,2);
\draw[<->,thick,>=stealth] (3,1) -- node[midway,above] {\tiny $P_4$} (4,1) -- (4,2);
\draw[<->,thick,>=stealth] (4,2) -- (5,2) node[above] {\tiny $P_3$} -- (5,1);
\draw[<->,thick,>=stealth] (4,1) -- node[midway,above] {\tiny $P_6$} (5,1);
\draw[<->,thick,>=stealth] (3,1) -- node[midway,above] {\tiny $P_5$} (2,1);
\draw[<->,thick,>=stealth] (4,2) -- node[midway,right] {\tiny $P_7$} (4,3);
\end{tikzpicture}
\caption{A maximally planar CPG graph.}
\label{pCPG}
\end{subfigure}
\caption{Two maximally planar graphs.}
\end{figure}


\section{Coloring CPG Graphs}
\label{sec:coloring}

In this section, we provide tight upper bounds on the chromatic number of $B_k$-CPG graphs for different values of $k$ and investigate the {\sc 3-Colorability} problem for CPG graphs. The proof of the following result is an easy exercise left to the reader (see Section \ref{app:6color} of the Appendix).

\begin{theorem}
\label{6color}
CPG graphs are 6-colorable.
\end{theorem}

\textit{Remark.} Since $K_6$ is $B_2$-CPG, this bound is tight for $B_k$-CPG graphs with $k \geq 2$. We leave as an open problem whether this bound is also tight for $B_1$-CPG graphs (note that it is at least 5 since $K_5$ is $B_1$-CPG).

\begin{theorem}
$B_0$-CPG graphs are 4-colorable. Moreover, $K_4$ is a 4-chromatic $B_0$-CPG graph.
\end{theorem}

\begin{proof}
Let $G$ be a $B_0$-CPG graph and $\mathcal{R}=(\mathcal{G},\mathcal{P})$ a $0$-bend CPG representation of $G$. Denote by $\mathcal{L}$ (resp. $\mathcal{C}$) the set of rows (resp. columns) of $\mathcal{G}$ on which lies at least one path of $\mathcal{P}$. Since the representation contains no bend, if $A$ is a row in $\mathcal{L}$ (resp. column in $\mathcal{C}$), then the set of vertices having their corresponding path on $A$ induces a collection of disjoint paths in $G$. If $B \neq A$ is another row in $\mathcal{L}$ (resp. column in $\mathcal{C}$), then no path in $A$ touches a path in $B$. Hence, it suffices to use two colors to color the vertices having their corresponding path in a row of $\mathcal{L}$ and two other colors to color the vertices having their corresponding path in a column of $\mathcal{C}$ to obtain a proper coloring of $G$. \qed
\end{proof}

It immediately follows from a result in \cite{Hlineny1bis} that the {\sc 3-colorability} problem is $\mathsf{NP}$-complete in CPG, even if the graph admits a representation in which each grid-point belongs to at most two paths. We conclude this section by a strenghtening of this result.

\begin{theorem}
\label{thm:3c}
{\sc 3-Colorability} is $\mathsf{NP}$-complete in $B_0$-CPG.
\end{theorem}

\begin{proof}
We exhibit a polynomial reduction from {\sc 3-Colorability} restricted to planar graphs of maximum degree 4, which was shown to be $\mathsf{NP}$-complete in~\cite{garey}.

Let $G=(V,E)$ be a planar graph of maximum degree 4. It follows from \cite{tamassia} that $G$ admits a grid embedding where each vertex is mapped to a grid-point and each edge is mapped to a grid-path with at most 4 bends, in such a way that all paths are interiorly disjoint (such an embedding can be obtained in linear time). Denote by $\mathcal{D} = (\mathcal{V}, \mathcal{E})$ such an embedding, where $\mathcal{V}$ is the set of grid-points in one-to-one correspondence with $V$ and $\mathcal{E}$ is the set of grid-paths in one-to-one correspondence with $E$. For any vertex $u \in V$, we denote by $(x_u,y_u)$ the grid-point in $\mathcal{V}$ corresponding to $u$ and by $P_u^N$ (resp. $P_u^S$) the path of $\mathcal{E}$, if any, having $(x_u,y_u)$ as an endpoint and using the grid-edge above (resp. below) $(x_u,y_u)$. For any edge $e \in E$, we denote by $P_e$ the path in $\mathcal{E}$ corresponding to $e$. We construct from $\mathcal{D}$ a 0-bend CPG representation $\mathcal{R}$ in such a way that the corresponding graph $G'$ is 3-colorable if and only if $G$ is 3-colorable.

By eventually adding rows and columns to the grid, we may assume that the interior of each path $P$ in $\mathcal{E}$ is surrounded by an empty region, i.e. no path $P' \neq P$ or grid-point of $\mathcal{V}$ lies in the interior of this region. In the following, we denote this region by $\mathcal{R}_P$ (delimited by red dashed lines in every subsequent figure) and assume, without loss of generality, that it is always large enough for the following operations. 

We first associate with every vertex $u \in V$ a \textit{vertical path} $P_u$ containing the grid-point $(x_u,y_u)$ as follows. If $P_u^N$ (resp. $P_u^S$) is not defined, the top (resp. lower) endpoint of $P_u$ is $(x_u,y_u + \varepsilon)$ (resp. $(x_u,y_u - \varepsilon)$) for a small enough $\varepsilon$ so that the segment $[(x_u,y_u),(x_u,y_u + \varepsilon)]$ (resp. $[(x_u,y_u),(x_u,y_u - \varepsilon)]$) touches no path of $\mathcal{E}$. If $P_u^N$ has at least one bend, then the top endpoint of $P_u$ lies at the border of $\mathcal{R}_{P_u^N}$ on column $x_u$ (see Fig. \ref{Pu1}). If $P_u^N$ has no bend, then the top endpoint of $P_u$ lies at the middle of $P_u^N$ (see Fig. \ref{Pu2}). Similarly, we define the lower endpoint of $P_u$ according to $P_u^S$: if  $P_u^S$ has at least one bend, then the lower endpoint of $P_u$ lies at the border of $\mathcal{R}_{P_u^S}$ on column $x_u$, otherwise it lies at the middle of $P_u^S$.

\begin{figure}[htb]
\centering
\begin{subfigure}[b]{.45\textwidth}
\centering
\begin{tikzpicture}
\node[circ,label={[label distance=.001cm]3:\scriptsize $(x_u,y_u)$}] (u) at (0,0) {};
\draw[-,thick] (u) -- (0,2) -- (1,2);
\draw[thick,dotted] (1,2) -- (1.5,2) node[right] {\scriptsize $P_u^N$};
\draw[-,thick] (u) -- (1,0);
\draw[thick,dotted] (1,0) -- (1.5,0);
\draw[-,thick] (u) -- (-1,0);
\draw[thick,dotted] (-1.5,0) -- (-1,0);
\draw[-,thick] (0,-1) -- (u);
\draw[thick,dotted] (0,-1.5) -- (0,-1);
\draw[thick,dashed,red] (-1.5,.4) -- (-.4,.4) -- (-.4,2.4) -- (1.5,2.4);
\draw[thick,dashed,red] (-1.5,-.4) -- (-.4,-.4) -- (-.4,-1.5);
\draw[thick,dashed,red] (.4,-1.5) -- (.4,-.4) -- (1.5,-.4);
\draw[thick,dashed,red] (1.5,.4) -- (.4,.4) -- (.4,1.6) -- (1.5,1.6);
\draw[->,thick,>=stealth,blue] (u) -- (0,2.4); 
\end{tikzpicture}
\caption{$P_u^N$ contains at least one bend.}
\label{Pu1}
\end{subfigure}
\hspace*{.5cm}
\begin{subfigure}[b]{.45\textwidth}
\centering
\begin{tikzpicture}
\node[circ,label={[label distance=.001cm]3:\scriptsize $(x_u,y_u)$}] (u) at (0,0) {};
\node[circ] (v) at (0,2) {};
\draw[-,thick] (u) -- (v) node[above right] {\scriptsize $P_u^N$};
\draw[-,thick] (u) -- (1,0);
\draw[thick,dotted] (1,0) -- (1.5,0);
\draw[-,thick] (u) -- (-1,0);
\draw[thick,dotted] (-1.5,0) -- (-1,0);
\draw[-,thick] (0,-1) -- (u);
\draw[thick,dotted] (0,-1.5) -- (0,-1);
\draw[thick,dashed,red] (-1.5,.4) -- (-.4,.4) -- (-.4,2);
\draw[thick,dashed,red] (1.5, .4) -- (.4,.4) -- (.4,2);
\draw[thick,dashed,red] (-1.5,-.4) -- (-.4,-.4) -- (-.4,-1.5);
\draw[thick,dashed,red] (.4,-1.5) -- (.4,-.4) -- (1.5,-.4);
\draw[->,thick,>=stealth,blue] (u) -- (0,1);
\end{tikzpicture}
\caption{$P_u^N$ contains no bend.}
\label{Pu2}
\end{subfigure}
\caption{Constructing the path $P_u$ corresponding to vertex $u$ (in blue).}
\label{pathcons}
\end{figure}
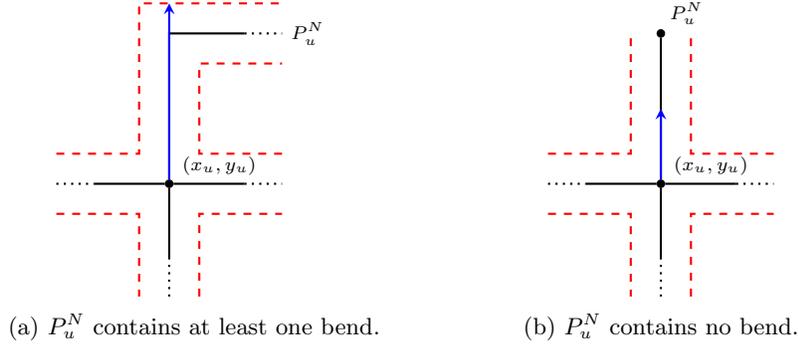

For any path $P$ of $\mathcal{E}$ with at least two bends, an \textit{interior vertical segment of $P$} is a vertical segment of $P$ containing none of its endpoints (note that since every path in $\mathcal{E}$ has at most 4 bends, it may contain at most two interior vertical segments). We next replace every interior segment of $P$ by a slightly longer vertical path touching the border of $\mathcal{R}_P$ (see Fig. \ref{intseg}). 

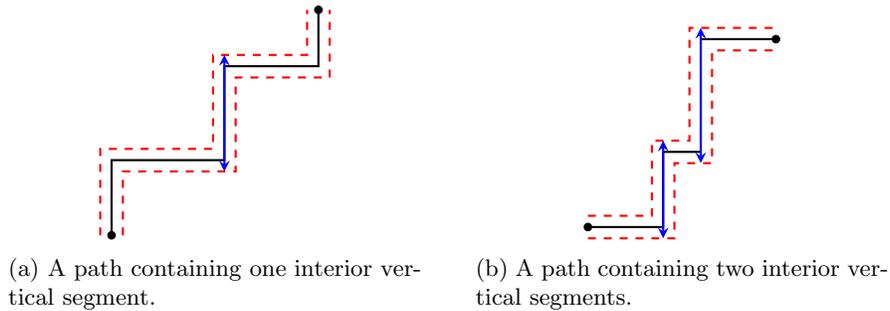
\begin{figure}
\centering
\begin{subfigure}[b]{.45\textwidth}
\centering
\begin{tikzpicture}[scale=.5]
\node[circ] at (0,0) {};
\node[circ] at (5.5,6) {};
\draw[-,thick] (0,0) -- (0,2) -- (3,2) -- (3,4.5) -- (5.5,4.5) -- (5.5,6);
\draw[thick,dashed,red] (-.3,0) -- (-.3,2.3) -- (2.7,2.3) -- (2.7,4.8) -- (5.2,4.8) -- (5.2,6);
\draw[thick,dashed,red] (.3,0) -- (.3,1.7) -- (3.3,1.7) -- (3.3,4.2) -- (5.8,4.2) -- (5.8,6);
\draw[<->,thick,>=stealth,blue] (3,1.7) -- (3,4.8);
\end{tikzpicture}
\caption{A path containing one interior vertical segment.}
\end{subfigure}
\hspace*{.5cm}
\begin{subfigure}[b]{.45\textwidth}
\centering
\begin{tikzpicture}[scale=.5]
\node[circ] at (0,0) {};
\node[circ] at (5,5) {};
\draw[-,thick] (0,0) -- (2,0) -- (2,2) -- (3,2) -- (3,5) -- (5,5);
\draw[thick,dashed,red] (0,.3) -- (1.7,.3) -- (1.7,2.3) -- (2.7,2.3) -- (2.7,5.3) -- (5,5.3);
\draw[thick,dashed,red] (0,-.3) -- (2.3,-.3) -- (2.3,1.7) -- (3.3,1.7) -- (3.3,4.7) -- (5,4.7);
\draw[<->,thick,>=stealth,blue] (2,-.3) -- (2,2.3);
\draw[<->,thick,>=stealth,blue] (3,1.7) -- (3,5.3);
\end{tikzpicture}
\caption{A path containing two interior vertical segments.}
\end{subfigure}
\caption{Replacing interior vertical segments by 0-bend paths (in blue).}
\label{intseg}
\end{figure}

We finally introduce two gadgets $H$ (see Fig. \ref{gadget}) and $H'$, where $H'$ is the subgraph of $H$ induced by $\{b,c,4,5,6,7,8,9,10\}$, as follows. Denote by $\mathcal{P}'$ the set of vertical paths introduced so far and by $V'$ the set of vertices of the contact graph of $\mathcal{P}'$. Observe that $V'$ contains a copy of $V$ and that two vertices are adjacent in the contact graph of $\mathcal{P}'$ if and only if they are both copies of vertices in $V$ and the path $P$ of $\mathcal{P}$ corresponding to the edge between these two copies is a vertical path with no bend. Now, along each path $P_{uv}$ of $\mathcal{P}$ such that the vertical paths $P_u$ and $P_v$ of $\mathcal{P}'$ do not touch, we add gadgets $H$ and $H'$ as follows. Let $P_1, \ldots, P_k$ be the vertical paths of $\mathcal{P}'$ encountered in order when going along $P_{uv}$ from $(x_u,y_u)$ to $(x_v,y_v)$ and let  $u_j$ be the vertex of $V'$ corresponding to $P_j$, for $1 \leq j \leq k$. Note that $P_1$ (resp. $P_k$) is the path corresponding to vertex $u= u_1$ (resp. $v=u_k$) and that $P_j$, for $2\leq j\leq k-1$, is a path corresponding to an interior vertical segment of $P_{uv}$ (this implies in particular that $k \leq 4$). We add the gadget $H'$ in between $u_1$ and $u_2$ by identifying $u_1$ with $b$ and $u_2$ with $c$. Moreover, for any $2 \leq j \leq k-1$, we add the gadget $H$ in between $u_j$ and $u_{j+1}$ by identifying $u_j$ with $b$ and $u_{j+1}$ with $a$ (see Fig. \ref{gadgetpath} where $k=4$ and each box labeled $H$ (resp. $H'$) means that gadget $H$ (resp. $H'$) has been added by identifying the vertex lying to the left of the box to $b$ and the vertex lying on the right of the box to $a$ (resp.~$c$)). 

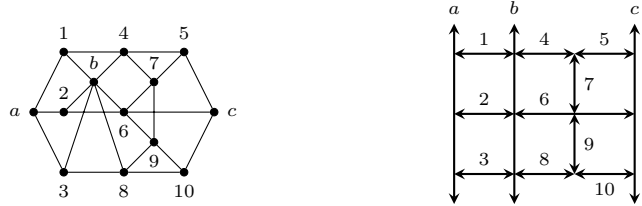
\begin{figure}
\centering
\begin{subfigure}[b]{.45\textwidth}
\centering
\begin{tikzpicture}[scale=.4]
\node[circ,label=left:{\scriptsize $a$}] (a) at (1,3) {};
\node[circ,label=below:{\scriptsize $3$}] (3) at (2,1) {};
\node[circ,label=above:{\scriptsize $2$}] (2) at (2,3) {};
\node[circ,label=above:{\scriptsize $1$}] (1) at (2,5) {};
\node[circ,label=above:{\scriptsize $b$}] (b) at (3,4) {};
\node[circ,label=below:{\scriptsize $8$}] (8) at (4,1) {};
\node[circ,label=below:{\scriptsize $6$}] (6) at (4,3) {};
\node[circ,label=above:{\scriptsize $4$}] (4) at (4,5) {};
\node[circ,label=below:{\scriptsize $9$}] (9) at (5,2) {};
\node[circ,label=above:{\scriptsize $7$}] (7) at (5,4) {};
\node[circ,label=below:{\scriptsize $10$}] (10) at (6,1) {};
\node[circ,label=above:{\scriptsize $5$}] (5) at (6,5) {};
\node[circ,label=right:{\scriptsize $c$}] (c) at (7,3) {};

\draw[-] (a) -- (1)
(a) -- (2)
(a) -- (3)
(1) -- (b)
(1) -- (4) 
(2) -- (b)
(2) -- (6)
(3) -- (b) 
(3) -- (8)
(b) -- (4)
(b) -- (6)
(b) -- (8)
(4) -- (5) 
(4) -- (7)
(6) -- (7) 
(6) -- (c)
(6) -- (9)
(8) -- (9)
(8) -- (10)
(7) -- (5)
(7) -- (9)
(9) -- (10)
(5) -- (c)
(10) -- (c);
\end{tikzpicture}
\end{subfigure}
\begin{subfigure}[b]{.45\textwidth}
\centering
\begin{tikzpicture}[scale=.8]
\draw[<->,thick,>=stealth] (1,1.5) -- (1,4.5) node[above] {\scriptsize $a$};
\draw[<->,thick,>=stealth] (1,2) -- (2,2) node[midway,above] {\scriptsize $3$};
\draw[<->,thick,>=stealth] (1,3) -- (2,3) node[midway,above] {\scriptsize $2$};
\draw[<->,thick,>=stealth] (1,4) -- (2,4) node[midway,above] {\scriptsize $1$};
\draw[<->,thick,>=stealth] (2,1.5) -- (2,4.5) node[above] {\scriptsize $b$};
\draw[<->,thick,>=stealth] (2,2) -- (3,2) node[midway,above] {\scriptsize $8$};
\draw[<->,thick,>=stealth] (2,3) -- (4,3) node[near start,above] {\scriptsize $6$};
\draw[<->,thick,>=stealth] (2,4) -- (3,4) node[midway,above] {\scriptsize $4$};
\draw[<->,thick,>=stealth] (3,2) -- (3,3) node[midway,right] {\scriptsize $9$};
\draw[<->,thick,>=stealth] (3,3) -- (3,4) node[midway,right] {\scriptsize $7$};
\draw[<->,thick,>=stealth] (3,2) -- (4,2) node[midway,below] {\scriptsize $10$};
\draw[<->,thick,>=stealth] (3,4) -- (4,4) node[midway,above] {\scriptsize $5$};
\draw[<->,thick,>=stealth] (4,1.5) -- (4,4.5) node[above] {\scriptsize $c$};
\end{tikzpicture}
\end{subfigure}
\caption{The gadget $H$ (left) and a $0$-bend CPG representation of it (right).}
\label{gadget}
\end{figure}

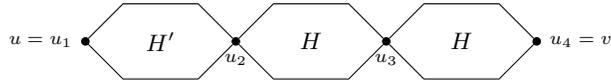
\begin{figure}[htb]
\centering
\begin{tikzpicture}
\node[circ,label=left:{\scriptsize $u=u_1$}] (u) at (1,0) {};
\node[circ,label=below:{\scriptsize $u_2$}] (u1) at (3,0) {};
\node[circ,label=below:{\scriptsize $u_3$}] (u2) at (5,0) {};
\node[circ,label=right:{\scriptsize $u_4=v$}] (u3) at (7,0) {};
\node[draw=none] at (2,0) {$H'$};
\node[draw=none] at (4,0) {$H$};
\node[draw=none] at (6,0) {$H$};

\draw[-] (u) -- (1.5,.5) -- (2.5,.5) -- (u1)
(u) -- (1.5,-.5) -- (2.5,-.5) -- (u1)
(u1) -- (3.5,.5) -- (4.5,.5) -- (u2)
(u1) -- (3.5,-.5) -- (4.5,-.5) -- (u2)
(u2) -- (5.5,.5) -- (6.5,.5) -- (u3)
(u2) -- (5.5,-.5) -- (6.5,-.5) -- (u3);
\end{tikzpicture}
\caption{Adding gadgets $H$ and $H'$.}
\label{gadgetpath}
\end{figure}

The resulting graph $G'$ remains $B_0$-CPG. Indeed, we may add $0$-bend CPG representations of the gadgets $H$ and $H'$ inside $\mathcal{R}_{P_{uv}}$ and at different heights so that they do not touch any other such gadget, as shown in Fig. \ref{addgad}. In Section \ref{app:ex} of the Appendix, we give a local example of the resulting $0$-bend CPG representation~$\mathcal{R}$.

\begin{figure}
\centering
\begin{subfigure}[b]{.45\textwidth}
\centering
\begin{tikzpicture}[scale=.8]
\draw[->,thick,>=stealth,blue] (1,-1) -- (1,2.1) node[near start,right] {\scriptsize $P_{u_j}$};
\draw[thick,dashed,red] (2,-1) -- (2,0) -- (6.2,0) -- (6.2,3);
\draw[<-,thick,>=stealth,blue] (5,0) --(5,3) node[near end,right] {\scriptsize $P_{u_{j+1}}$};
\draw[<->,thick,>=stealth] (.2,1.1) -- (.2,1.9);
\draw[<->,thick,>=stealth] (.6,1.1) -- (.6,1.5);
\draw[<->,thick,>=stealth] (.6,1.5) -- (.6,1.9);
\draw[<->,thick,>=stealth] (.2,1.1) -- (.6,1.1);
\draw[<->,thick,>=stealth] (.6,1.1) -- (1,1.1);
\draw[<->,thick,>=stealth] (.2,1.5) -- (1,1.5);
\draw[<->,thick,>=stealth] (.2,1.9) -- (.6,1.9);
\draw[<->,thick,>=stealth] (.6,1.9) -- (1,1.9);
\draw[thick,dashed,red] (0,-1) -- (0,2.1) -- (4,2.1) -- (4,3);
\draw[<->,thick,>=stealth] (1,1.1) -- (5,1.1);
\draw[<->,thick,>=stealth] (1,1.5) -- (5,1.5);
\draw[<->,thick,>=stealth] (1,1.9) -- (5,1.9);
\end{tikzpicture}
\caption{Adding gadget $H$.}
\label{addH}
\end{subfigure}
\hspace*{.5cm}
\begin{subfigure}[b]{.45\textwidth}
\centering
\begin{tikzpicture}[scale=.8]
\draw[->,thick,>=stealth,blue] (1,-1) -- (1,2.1) node[near start,right] {\scriptsize $P_u$};
\draw[thick,dashed,red] (2,-1) -- (2,0) -- (6,0) -- (6,3);
\draw[<-,thick,>=stealth,blue] (5,0) --(5,3) node[near end,right] {\scriptsize $P_{u_2}$};
\draw[thick,dashed,red] (0,-1) -- (0,2.1) -- (4,2.1) -- (4,3);
\draw[<->,thick,>=stealth] (1,1.1) -- (3,1.1);
\draw[<->,thick,>=stealth] (3,1.1) -- (5,1.1);
\draw[<->,thick,>=stealth] (1,1.5) -- (5,1.5);
\draw[<->,thick,>=stealth] (1,1.9) -- (3,1.9);
\draw[<->,thick,>=stealth] (3,1.9) -- (5,1.9);
\draw[<->,thick,>=stealth] (3,1.1) -- (3,1.5);
\draw[<->,thick,>=stealth] (3,1.5) -- (3,1.9);
\end{tikzpicture}
\caption{Adding gadget $H'$.}
\label{addH'}
\end{subfigure}
\caption{Locally adding gadgets to control the color of the vertices.}
\label{addgad}
\end{figure}
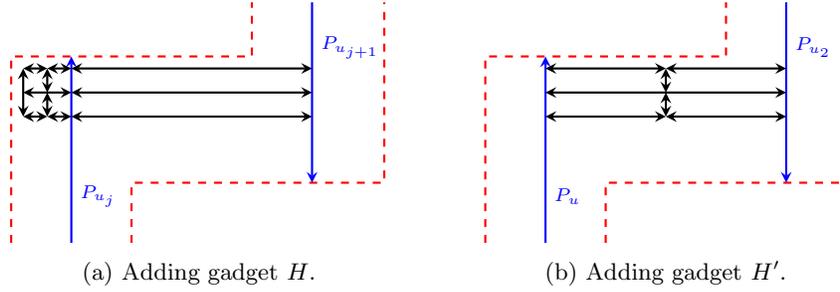

We now show that $G$ is 3-colorable if and only if $G'$ is. To this end, we prove the following.

\begin{Claim}
\label{claim1}
\begin{minipage}[t]{\linegoal}
\begin{itemize}[leftmargin=*]
\item[$\bullet$] In any 3-coloring $\mathbf{c}$ of $H'$, we have $\mathbf{c}(b) \neq \mathbf{c}(c)$.
\item[$\bullet$] In any 3-coloring $\mathbf{c}$ of $H$, we have $\mathbf{c}(a) = \mathbf{c}(b)$ and $\mathbf{c}(b) \neq \mathbf{c}(c)$.\\
\end{itemize}
\end{minipage}
\end{Claim}

\begin{cproof}
Let $\mathbf{c} \colon \{a,b,c,1,2,3,4,5,6,7,8,9,10\} \to \{blue,red,green\}$ be a 3-coloring of $H$ and assume without loss of generality that $\mathbf{c}(b) = blue$. Clearly, at least two vertices among $4$, $6$ and $8$ have the same color. If vertices $4$, $6$ and $8$ all have the same color, say $red$, then either $\mathbf{c}(7) = blue$ and $\mathbf{c}(9) = green$, or $\mathbf{c}(7) = green$ and $\mathbf{c}(9) = blue$. Therefore, $\{\mathbf{c}(5),\mathbf{c}(10)\} = \{blue,green\}$ and since $c$ is adjacent to all three colors, we then obtain a contradiction. Now if vertices $4$ and $8$ have the same color, say $red$, then vertex $6$ has color $green$ and both $7$ and $9$ have color $blue$, a contradiction. Hence, either $\mathbf{c}(4) = \mathbf{c}(6) \neq \mathbf{c}(8)$ or $\mathbf{c}(8) = \mathbf{c}(6) \neq \mathbf{c}(4)$. By symmetry, we may assume that vertices $4$ and $6$ have the same color, say $red$, and that vertex $8$ has color $green$. This implies that vertex $7$ has color $green$, vertices $9$ and $5$ have color $blue$ and vertex $10$ has color $red$; but then, $\mathbf{c}(c) = green \neq \mathbf{c}(b)$. This proves the first point of the claim. Observe that each coloring of $b$ and $c$ with distinct colors can be extended to a 3-coloring of $H'$ and $H$.  

As for the second point, since vertices $4$ and $6$ have color $red$, both $1$ and $2$ must have color $green$, and since vertex $8$ has color $green$, vertex $3$ must have color $red$. Consequently, $\mathbf{c}(a) = blue = \mathbf{c}(b)$.
\end{cproof}\\

We finally conclude the proof of Theorem \ref{thm:3c}. By Claim \ref{claim1}, if $\mathbf{c}$ is a 3-coloring of $G'$ then, for any path $P_{uv}$ of $\mathcal{P}$, we have $\mathbf{c}(u_1) \neq \mathbf{c}(u_2)$ and $\mathbf{c}(u_2) =  \mathbf{c}(u_i)$ for all $3\leq i \leq k$. Hence, $\mathbf{c}$ induces a 3-coloring of $G$. Conversely, it is easy to see that any 3-coloring of $G$ can be extended to a 3-coloring of $G'$. \qed
\end{proof}


\section{Conclusion}
\label{sec:conclusion}

We conclude by stating the following open questions:

\begin{enumerate}
\item[1.] Are $B_1$-CPG graphs 5-colorable?
\item[2.] Can we characterize those planar graphs which are CPG?
\item[3.] Is {\sc Recognition} $\mathsf{NP}$-complete for $B_k$-CPG graphs with $k >0$ ?
\end{enumerate}

 
 \bibliographystyle{plain}
\bibliography{references}

\begin{thebibliography}{10}

\bibitem{felsner1}
N.~Aerts and S.~Felsner.
\newblock Vertex contact representations of paths on a grid.
\newblock In {\em Proceedings of the 40th International Workshop in
  Graph-Theoretic Concepts in Computer Science}, WG '14, pages 56--68. Springer
  Berlin Heidelberg, 2014.

\bibitem{alcon}
L.~Alc{\'o}n, F.~Bonomo, and M.~P. Mazzoleni.
\newblock Vertex intersection graphs of paths on a grid: Characterization
  within block graphs.
\newblock {\em Graphs and Combinatorics}, 33(4):653--664, 2017.

\bibitem{asinowski}
A.~Asinowski, E.~Cohen, M.~C. Golumbic, V.~Limouzy, M.~Lipshteyn, and M.~Stern.
\newblock Vertex intersection graphs of paths on a grid.
\newblock {\em Journal of Graph Algorithms and Applications}, 16:129--150,
  2012.

\bibitem{beineke}
L.~W. Beineke.
\newblock Characterizations of derived graphs.
\newblock {\em Journal of Combinatorial Theory}, 9(2):129--135, 1970.

\bibitem{chaplick}
S.~Chaplick, E.~Cohen, and J.~Stacho.
\newblock Recognizing some subclasses of vertex intersection graphs of 0-bend
  paths in a grid.
\newblock In {\em Proceedings of the 37th International Workshop in
  Graph-Theoretic Concepts in Computer Science}, WG '09, pages 319--330.
  Springer Berlin Heidelberg, 2009.

\bibitem{chaplick13}
S.~Chaplick, S.~G. Kobourov, and T.~Ueckerdt.
\newblock Equilateral {L}-contact graphs.
\newblock In {\em Proceedings of the 39th International Workshop in
  Graph-Theoretic Concepts in Computer Science}, WG '13, pages 139--151.
  Springer Berlin Heidelberg, 2013.

\bibitem{chaplick12}
S.~Chaplick and T.~Ueckerdt.
\newblock Planar graphs as {VPG}-graphs.
\newblock In {\em Lecture Notes in Computer Science vol. 7704, Graph Drawing},
  GD '12, pages 174--186. Springer Berlin Heidelberg, 2012.

\bibitem{cohen1}
E.~Cohen, M.~C. Golumbic, and B.~Ries.
\newblock Characterizations of cographs as intersection graphs of paths on a
  grid.
\newblock {\em Discrete Applied Mathematics}, 178:46--57, 2014.

\bibitem{cohen2}
E.~Cohen, M.~C. Golumbic, W.~T. Trotter, and R.~Wang.
\newblock Posets and {VPG} graphs.
\newblock {\em Order}, 33(1):39--49, 2016.

\bibitem{castro}
N.~de~Castro, F.~J. Cobos, J.~C. Dana, A.~M{\'a}rquez, and M.~Noy.
\newblock Triangle-free planar graphs as segments intersection graphs.
\newblock In {\em Lecture Notes in Computer Science vol. 1731, Graph Drawing},
  GD '99, pages 341--350. Springer Berlin Heidelberg, 1999.

\bibitem{Fraysseix1}
H.~de~Fraysseix and P.~O. de~Mendez.
\newblock Representations by contact and intersection of segments.
\newblock {\em Algorithmica}, 47(4):453--463, 2007.

\bibitem{diestel}
R.~Diestel.
\newblock {\em Graph Theory, 4th Edition}, volume 173 of {\em Graduate texts in
  mathematics}.
\newblock Springer, 2012.

\bibitem{Felsner}
S.~Felsner, K.~Knauer, G.~B. Mertzios, and T.~Ueckerdt.
\newblock Intersection graphs of {L}-shapes and segments in the plane.
\newblock {\em Discrete Applied Mathematics}, pages 48--55, 2016.

\bibitem{francis}
M.~Francis and A.~Lahiri.
\newblock {VPG} and {EPG} bend-numbers of {Halin} graphs.
\newblock {\em Discrete Applied Mathematics}, 215:95--105, 2016.

\bibitem{garey}
M.R. Garey, D.S. Johnson, and L.~Stockmeyer.
\newblock Some simplified $\mathsf{NP}$-complete graph problems.
\newblock {\em Theoretical Computer Science}, 1(3):237 -- 267, 1976.

\bibitem{garg}
A.~Garg and R.~Tamassia.
\newblock On the computational complexity of upward and rectilinear planarity
  testing.
\newblock {\em SIAM J. Comput.}, 2001.

\bibitem{Golumbic1}
M.~C. Golumbic and B.~Ries.
\newblock On the intersection graphs of orthogonal line segments in the plane:
  Characterizations of some subclasses of chordal graphs.
\newblock {\em Graphs and Combinatorics}, 29(3):499--517, 2013.

\bibitem{gonca}
D.~Gon\c{c}alves, L.~Isenmann, and C.~Pennarun.
\newblock Planar graphs as {L}-intersection or {L}-contact graphs.
\newblock In {\em Proceedings of the 29th Annual ACM-SIAM Symposium on Discrete
  Algorithms}, SODA '18, pages 172--184. Society for Industrial and Applied
  Mathematics, 2018.

\bibitem{heldt1}
D.~Heldt, K.~Knauer, and T.~Ueckerdt.
\newblock On the bend-number of planar and outerplanar graphs.
\newblock {\em Discrete Applied Mathematics}, 179(C):109--119, 2014.

\bibitem{Hlineny}
P.~Hlin\v{e}n\'{y}.
\newblock Classes and recognition of curve contact graphs.
\newblock {\em Journal of Combinatorial Theory, Series B}, 74(1):87--103, 1998.

\bibitem{Hlineny1bis}
P.~Hlin\v{e}n\'{y}.
\newblock The maximal clique and colourability of curve contact graphs.
\newblock {\em Discrete Applied Mathematics}, 81(1):59 -- 68, 1998.

\bibitem{Hlineny1}
P.~Hlin\v{e}n\'{y}.
\newblock Contact graphs of line segments are $\mathsf{NP}$-complete.
\newblock {\em Discrete Mathematics}, 235(1-3):95--106, 2001.

\bibitem{4reg}
Lehel J.
\newblock Generating all 4-regular planar graphs from the graph of the
  octahedron.
\newblock {\em Journal of Graph Theory}, 5(4):423--426.

\bibitem{kabourov}
S.~Kobourov, T.~Ueckerdt, and K.~Verbeek.
\newblock Combinatorial and geometric properties of planar {Laman} graphs.
\newblock In {\em Proceedings of the 24th Annual ACM-SIAM Symposium on Discrete
  Algorithms}, SODA '13, pages 1668--1678. Society for Industrial and Applied
  Mathematics, 2013.

\bibitem{korzhik}
V.~P. Korzhik.
\newblock Minimal non-1-planar graphs.
\newblock {\em Discrete Mathematics}, 308(7):1319--1327, 2008.

\bibitem{tamassia}
R.~Tamassia and I.~G. Tollis.
\newblock Planar grid embedding in linear time.
\newblock {\em IEEE Transactions on Circuits and Systems}, 36(9):1230--1234,
  1989.

\end{thebibliography}

\newpage

\section{Appendix}

\subsection{Proof of the existence of an infinite family of 6-regular CPG graphs}
\label{app:6reg}

It is clear that there exists an infinite family of CPG graphs having a vertex of degree at most 5. On the other hand, the existence of an infinite family of 6-regular CPG graphs is a priori not guaranteed. We can however show that it is the case. Indeed, consider a 4-regular planar graph $G$. From \cite{tamassia}, it follows that $G$ admits an embedding on the grid where each vertex is mapped to a distinct grid-point and each edge $e$ is mapped to a path on the grid whose endpoints are the grid-points corresponding to the endvertices of $e$, in such a manner that all paths are interiorly disjoint. We derive therefrom a CPG representation of the line graph $L(G)$ of $G$ which is 6-regular: each edge $e$ in the embedding of $G$ on the grid corresponds to the path associated with vertex $v_e$ of $L(G)$ and each vertex in the embedding of $G$ on the grid is the grid-point where the four corresponding paths pairwise touch. The existence of an infinite family of 6-regular CPG graphs then follows from the existence of an infinite family of 4-regular planar graphs~\cite{4reg}.

\subsection{4-regular graphs on 7 vertices are non-planar}
\label{4-regular}

We show that every 4-regular graph $G=(V,E)$ on 7 vertices contains $K_{3,3}$ as a minor. Let $v_1 \in V$ and $v_2,v_3 \notin N(v_1)$. Then $v_2$ and $v_3$ have at least 3 neighbors in $N(v_1)$. If $N(v_2) \cap N(v_1) \subseteq N(v_3)$, then clearly $G$ contains $K_{3,3}$ as a minor. Hence, $v_2$ has a neighbor $x \in N(v_1)$ which is non-adjacent to $v_3$; by symmetry, $v_3$ also has a neighbor $y \in N(v_1)$ which is non-adjacent to $v_2$. But then $v_2$ and $v_3$ must be adjacent as well as $x$ and $y$ (recall that the graph is 4-regular), and again $G$ contains $K_{3,3}$ as a minor.

\subsection{Proof of Theorem \ref{6color}}
\label{app:6color}

If $G$ is a CPG graph, by Lemma \ref{degree} $G$ is either 6-regular or has a vertex of degree at most 5. In the first case, the result follows from Theorem \ref{K7} and Brooks' Theorem (any graph $H$, apart from the complete graph and the cycle of odd length, may be colored using $\Delta$ colors where $\Delta$ is the maximum degree of $H$). Otherwise, $G$ contains a vertex of degree at most 5 and we conclude by induction.

\subsection{Proof of Proposition \ref{K6}}
\label{app:k6}
Before turning to the proof of Proposition \ref{K6}, we first make several observations regarding $B_1$-CPG graphs. All following statements remain true up to reflection across a vertical or horizontal line and by inverting the role of rows and columns.

\begin{observation}
Let $G$ be a $B_1$-CPG graph and $\mathcal{R}=(\mathcal{G},\mathcal{P})$ be a $1$-bend CPG representation of $G$. Assume there exist two distinct grid-points in $\mathcal{G}$, $p = (x_j,y_k)$ and $p' = (x_i,y_l)$, with $j < i$, such that $p$ (resp. $p'$) is an endpoint or the bend-point of a path $P$ (resp. $P'$) in $\mathcal{P}$. Then, if $P$ (resp. $P'$) uses the grid-edge on the left of $p$ (resp. on the right of $p'$), $P$ and $P'$ can not touch (see Fig. \ref{fobs1}). 
\label{obs1}
\end{observation}

\begin{figure}[ht]
\centering
\begin{minipage}{.3\textwidth}
\centering
\begin{tikzpicture}[scale=0.7]
\draw[dotted] (3,1)--(0,1) node[left] {\tiny $y_l$};
\draw[dotted] (3,2)--(0,2) node[left] {\tiny $y_k$};

\draw[dotted] (1,3)--(1,0) node[below] {\tiny $x_j$};
\draw[dotted] (2,3)--(2,0) node[below] {\tiny $x_i$};

\draw[thick,->,>=stealth] (0.5,2)--node[midway,above] {\tiny $P$}(1,2);

\draw[thick,<-,>=stealth] (2,1)--node[above] {\tiny $P'$}(2.5,1);	
\end{tikzpicture}
\end{minipage}
\begin{minipage}{.3\textwidth}
\centering
\begin{tikzpicture}[scale=0.7]
\draw[dotted] (3,1)--(0,1) node[left] {\tiny $y_l$};
\draw[dotted] (3,2)--(0,2) node[left] {\tiny $y_k$};

\draw[dotted] (1,3)--(1,0) node[below] {\tiny $x_j$};
\draw[dotted] (2,3)--(2,0) node[below] {\tiny $x_i$};

\draw[thick,-,>=stealth] (0.5,2)--node[midway,above] {\tiny $P$}(1,2)--(1,1.5);

\draw[thick,<-,>=stealth] (2,1)--node[above] {\tiny $P'$}(2.5,1);	
\end{tikzpicture}
\end{minipage}
\begin{minipage}{.3\textwidth}
\centering
\begin{tikzpicture}[scale=0.7]
\draw[dotted] (3,1)--(0,1) node[left] {\tiny $y_l$};
\draw[dotted] (3,2)--(0,2) node[left] {\tiny $y_k$};

\draw[dotted] (1,3)--(1,0) node[below] {\tiny $x_j$};
\draw[dotted] (2,3)--(2,0) node[below] {\tiny $x_i$};

\draw[thick,-,>=stealth] (0.5,2)--node[midway,above] {\tiny $P$}(1,2)--(1,1.5);

\draw[thick,-,>=stealth] (2.5,1)--node[midway,above] {\tiny $P'$}(2,1)--(2,0.5);	
\end{tikzpicture}
\end{minipage}
\caption{Examples where $P$ and $P'$ can not touch.}
\label{fobs1}
\end{figure}
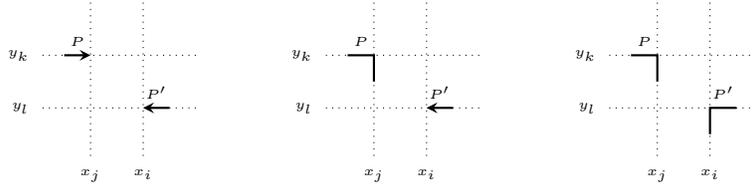

\begin{observation}
Let $G$ be a $B_1$-CPG graph and $\mathcal{R}=(\mathcal{G},\mathcal{P})$ be a $1$-bend CPG representation of $G$. Assume there exist three distinct grid-points in $\mathcal{G}$, $p = (x_i,y_t)$, $p' = (x_j,y_s)$ and $p'' = (x_k,y_r)$, with $k < j < i$, such that $p$ (resp. $p'$) is an endpoint of a path $P$ (resp. $P'$) in $\mathcal{P}$ and $p''$ is an endpoint or the bend-point of a path $P''$ in $\mathcal{P}$. Then, if $P''$ (resp. $P$, $P'$) uses the grid-edge on the left of $p''$ (resp. below $p$, below $p'$), all three paths can not pairwise touch (see Fig. \ref{fobs2}). 
\label{obs2}
\end{observation}

\begin{figure}[ht]
\centering
\begin{subfigure}[b]{.45\textwidth}
\centering
\begin{tikzpicture}[scale=0.5]
\draw[dotted] (4,1)--(0,1) node[left] {\tiny $y_t$};
\draw[dotted] (4,2)--(0,2) node[left] {\tiny $y_s$};
\draw[dotted] (4,3)--(0,3) node[left] {\tiny $y_r$};

\draw[dotted] (1,4)--(1,0) node[below] {\tiny $x_k$};
\draw[dotted] (2,4)--(2,0) node[below] {\tiny $x_j$};
\draw[dotted] (3,4)--(3,0) node[below] {\tiny $x_i$};

\draw[thick,->,>=stealth] (.5,3)--node[above] {\tiny $P''$}(1,3);
\draw[thick,->,>=stealth] (2,1.5)--node[right] {\tiny $P'$}(2,2);

\draw[thick,->,>=stealth] (3,.5)--node[right] {\tiny $P$}(3,1);
\end{tikzpicture}
\end{subfigure}
\hspace*{1cm}
\begin{subfigure}[b]{.45\textwidth}
\centering
\begin{tikzpicture}[scale=0.5]
\draw[dotted] (4,1)--(0,1) node[left] {\tiny $y_t$};
\draw[dotted] (4,2)--(0,2) node[left] {\tiny $y_s$};
\draw[dotted] (4,3)--(0,3) node[left] {\tiny $y_r$};

\draw[dotted] (1,4)--(1,0) node[below] {\tiny $x_k$};
\draw[dotted] (2,4)--(2,0) node[below] {\tiny $x_j$};
\draw[dotted] (3,4)--(3,0) node[below] {\tiny $x_i$};

\draw[thick,-,>=stealth] (.5,3)--node[above] {\tiny $P''$}(1,3)--(1,2.5);
\draw[thick,->,>=stealth] (2,1.5)--node[right] {\tiny $P'$}(2,2);

\draw[thick,->,>=stealth] (3,.5)--node[right] {\tiny $P$}(3,1);
\end{tikzpicture}
\end{subfigure}
\caption{Examples where $P$, $P'$ and $P''$ can not pairwise touch.}
\label{fobs2}
\end{figure}
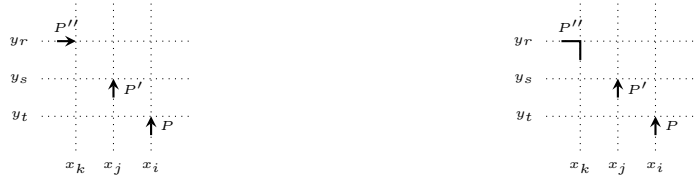

\textit{Proof of Proposition \ref{K6}.}
$K_6$ is in $B_2$-CPG as shown in Fig. \ref{K6B2}. Assume by contradiction that $K_6$ is a $B_1$-CPG graph and consider a $1$-bend CPG representation $\mathcal{R}=(\mathcal{G},\mathcal{P})$ of $K_6$. Since every vertex is of degree 5, we can assume, without loss of generality, that every endpoint of a path belongs to another path.
 
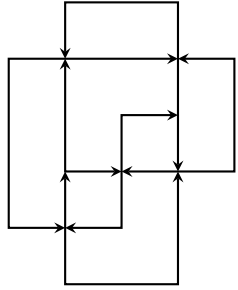
\begin{figure}[htb]
\centering   
\begin{tikzpicture}[scale=.75]
\draw[thick,<->,>=stealth] (1,1)--(0,1)--(0,4)--(3,4);
\draw[thick,<->,>=stealth] (1,1)--(2,1)--(2,3)--(3,3);
\draw[thick,<->,>=stealth] (1,4)--(1,2)--(2,2);
\draw[thick,<->,>=stealth] (1,4)--(1,5)--(3,5)--(3,2);
\draw[thick,<->,>=stealth] (1,2)--(1,0)--(3,0)--(3,2);
\draw[thick,<->,>=stealth] (2,2)--(4,2)--(4,4)--(3,4);
\end{tikzpicture}
\caption{A $2$-bend CPG representation of $K_6$.}
\label{K6B2}
\end{figure} 

In the following, let $P_a$ with $a \in V(K_6)$, be a path in $\mathcal{P}$ and denote by $p_1$ and $p_2$ its two endpoints.

\begin{Claim}
\label{c1}
$\mathcal{R}$ contains no grid-point of type I.
\end{Claim} 

\begin{cproof} Assume, without loss of generality, that $p_1$ is a grid-point of type I. Since $d(a) = 5$, there exists a path $P_b$ either touching $P_a$ in its interior or having $p_2$ as an endpoint. If $P_a$ has no bend, it follows from Observation \ref{obs1} that $P_b$ can not touch one of the paths touching $P_a$ at $p_1$. If $P_a$ has a bend, we may assume without loss of generality that the horizontal segment of $P_a$ contains $p_1$ and lies to the left and below its vertical segment (see Fig. \ref{fc1}). Then, either we conclude similarly by Observation \ref{obs1} (see Fig. \ref{fc1a} where colors are used to specify which paths cannot touch depending on the position of $P_b$); or, denote by $P_c$ the other path not touching $P_a$ at $p_1$. If $P_b$ has $p_2$ as an endpoint, then by Observation \ref{obs1}, $P_b$ must use the grid-edge on the left of $p_2$ and $P_c$ must touch the interior of the vertical segment of $P_a$ and lie to its left; and we conclude by Observation \ref{obs2} that $P_b$, $P_c$ and the path with endpoint $p_1$ using the grid-edge below $p_1$ can then not pairwise touch (see Fig. \ref{fc1b} where colors are used to specify which paths cannot pairwise touch). 

\begin{figure}[ht]
\begin{subfigure}[b]{.45\textwidth}
\centering
\begin{tikzpicture}[scale=0.7]
\node[invisible, label=below right:{\tiny $P_a$}] (Pa) at (0,0) {};

\draw[<-,thick,>=stealth] (-2,0)--(0,0);
\draw[->,thick,>=stealth] (0,0)--(0,2);
\draw[->,thick,red,>=stealth] (-3,0)--(-1.98,0);
\draw[<-,thick,red,>=stealth] (0,1)--(1,1) node[right] {\tiny $P_b$};
\draw[<-,thick,red,>=stealth] (0,2) -- (1,2) node[right] {\tiny $P_b$};
\draw[->,thick,blue,>=stealth] (-2,1)--(-2,0.02);
\draw[<-,thick,blue,>=stealth] (-0.6,0)--(-0.6,-1) node[midway,left] {\tiny $P_b$};
\draw[->,thick,green,>=stealth] (-2,-1)--(-2,-0.02);
\draw[<-,thick,green,>=stealth] (-1.2,0)--(-1.2,1) node[above] {\tiny $P_b$};
\draw[->,thick,green,>=stealth] (0,3) -- (0,2) node[midway,left] {\tiny $P_b$};
\end{tikzpicture}
\caption{Different positions for $P_b$.}
\label{fc1a}
\end{subfigure}
\hspace*{1cm}
\begin{subfigure}[b]{.45\textwidth}
\centering
\begin{tikzpicture}[scale=0.7]
\node[invisible, label=below right:{\tiny $P_a$}] (Pa) at (0,0) {};

\draw[<-,thick,>=stealth] (-2,0)--(0,0);
\draw[-,thick,>=stealth] (0,0)--(0,2);
\draw[->,thick,>=stealth] (-3,0)--(-1.98,0);
\draw[->,thick,>=stealth] (-2,1)--(-2,0.02);
\draw[->,thick,blue,>=stealth] (-2,-1)--(-2,-0.02);
\draw[<-,thick,blue,>=stealth] (0,1)--(-1,1) node[left] {\tiny $P_c$};
\draw[<-,thick,blue,>=stealth] (0,2)--(-1,2) node[left] {\tiny $P_b$};
\end{tikzpicture}
\caption{Three paths that can not pairwise touch.}
\label{fc1b}
\end{subfigure}
\caption{$P_a$ has a bend.}
\label{fc1}
\end{figure}
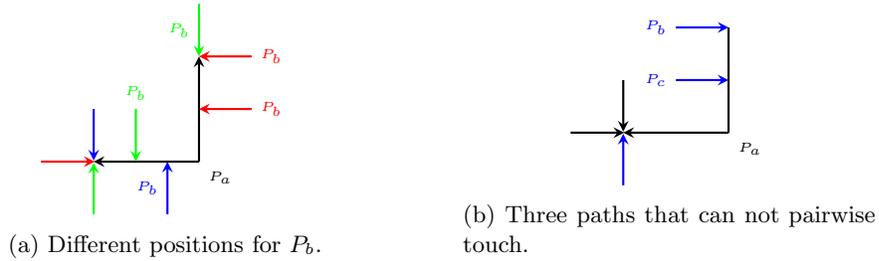

Otherwise, $P_b$ touches the interior of the vertical segment of $P_a$ while lying to its left, and $P_c$ contains $p_2$; but then, $P_c$ has an endpoint $p$ lying to the right of the vertical segment of $P_a$ and we distinguish three cases depending on the row on which is $p$. First, if $p$ lies on the same row as $p_2$, we conclude by Observation \ref{obs2} that three $P_b$, $P_c$ and the path with endpoint $p_1$ using the grid-edge below $p_1$ cannot pairwise touch (see Fig. \ref{figc1a} where colors are used to specify which paths cannot pairwise touch). Second, if $P_c$ has a bend and $p$ lies on a row above the row of $p_2$, then $P_c$ cannot touch the path with endpoint $p_1$ using the grid-edge below $p_1$ (see Fig. \ref{figc1b} where colors are used to specify which paths cannot touch). Finally, if $P_c$ has a bend and $p$ lies on a row below the row of $p_2$, then $P_c$, $P_b$ and the path with endpoint $p_1$ using the grid-edge on the left of $p_1$ cannot pairwise touch (see Fig. \ref{figc1c} where colors are used to specify which paths cannot pairwise touch), which concludes the proof. 
\end{cproof}

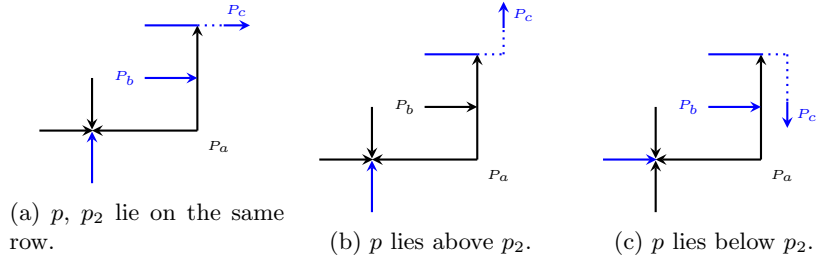
\begin{figure}[ht]
\centering
\begin{subfigure}[b]{.3\textwidth}
\centering
\begin{tikzpicture}[scale=0.7]
\node[invisible, label=below right:{\tiny $P_a$}] (Pa) at (0,0) {};

\draw[<-,thick,>=stealth] (-2,0)--(0,0);
\draw[->,thick,>=stealth] (0,0)--(0,2);
\draw[->,thick,>=stealth] (-3,0)--(-1.98,0);
\draw[->,thick,>=stealth] (-2,1)--(-2,0.02);
\draw[->,thick,>=stealth,blue] (-2,-1)--(-2,-0.02);
\draw[<-,thick,blue,>=stealth] (0,1)--(-1,1) node[left] {\tiny $P_b$};
\draw[-,thick,blue,>=stealth] (0,2)--(-1,2);
\draw[dotted,thick,blue,>=stealth] (0,2)--(.5,2);
\draw[->,thick,blue,>=stealth] (.5,2)--(1,2) node[midway,above] {\tiny $P_c$};
\end{tikzpicture}
\caption{$p$, $p_2$ lie on the same row.}
\label{figc1a}
\end{subfigure}
\begin{subfigure}[b]{.3\textwidth}
\centering
\begin{tikzpicture}[scale=0.7]
\node[invisible, label=below right:{\tiny $P_a$}] (Pa) at (0,0) {};

\draw[<-,thick,>=stealth] (-2,0)--(0,0);
\draw[->,thick,>=stealth] (0,0)--(0,2);
\draw[->,thick,>=stealth] (-3,0)--(-1.98,0);
\draw[->,thick,>=stealth] (-2,1)--(-2,0.02);
\draw[->,thick,>=stealth,blue] (-2,-1)--(-2,-0.02);
\draw[<-,thick,>=stealth] (0,1)--(-1,1) node[left] {\tiny $P_b$};
\draw[-,thick,blue,>=stealth] (0,2)--(-1,2);
\draw[dotted,thick,blue,>=stealth] (0,2)--(.5,2);
\draw[dotted,thick,blue,>=stealth] (.5,2)--(.5,2.5);
\draw[->,thick,blue,>=stealth] (.5,2.5)--(.5,3) node[midway,right] {\tiny $P_c$};
\end{tikzpicture}
\caption{$p$ lies above $p_2$.}
\label{figc1b}
\end{subfigure}
\begin{subfigure}[b]{.3\textwidth}
\centering
\begin{tikzpicture}[scale=0.7]
\node[invisible, label=below right:{\tiny $P_a$}] (Pa) at (0,0) {};

\draw[<-,thick,>=stealth] (-2,0)--(0,0);
\draw[->,thick,>=stealth] (0,0)--(0,2);
\draw[->,thick,>=stealth,blue] (-3,0)--(-1.98,0);
\draw[->,thick,>=stealth] (-2,1)--(-2,0.02);
\draw[->,thick,>=stealth] (-2,-1)--(-2,-0.02);
\draw[<-,thick,>=stealth,blue] (0,1)--(-1,1) node[left] {\tiny $P_b$};
\draw[-,thick,blue,>=stealth] (0,2)--(-1,2);
\draw[dotted,thick,blue,>=stealth] (0,2)--(.5,2);
\draw[dotted,thick,blue,>=stealth] (.5,2)--(.5,1.1);
\draw[->,thick,blue,>=stealth] (.5,1.1)--(.5,.6) node[midway,right] {\tiny $P_c$};
\end{tikzpicture}
\caption{$p$ lies below $p_2$.}
\label{figc1c}
\end{subfigure}
\caption{An endpoint of $P_c$ lies to the right of the vertical segment of $P_a$.}
\end{figure}

\begin{Claim}
\label{c2}
$\mathcal{R}$ contains no grid-point of type II.
\end{Claim}

\begin{cproof} We first show that both endpoints of any path in $\mathcal{P}$ cannot be grid-points of type II. For the sake of contradiction, assume, without loss of generality, that both $p_1$ and $p_2$ are grid-points of type II. If $P_a$ has no bend, then by Observation \ref{obs1}, one path touching $P_a$ at $p_1$ and one path touching $P_a$ at $p_2$ cannot touch (see Fig. \ref{samesub} where colors are used to specify which paths cannot touch). 

\begin{figure}[ht]
\centering
\begin{subfigure}[b]{.3\textwidth}
\centering
\begin{tikzpicture}
\draw[<->,thick,>=stealth] (-1,0) -- node[invisible,label=above:{\tiny $P_a$},midway] {} (1,0);
\draw[->,thick,>=stealth,blue] (-1.5,0) -- (-1,0);
\draw[<-,thick,>=stealth,blue] (1,0) -- (1.5,0);
\draw[-,thick,>=stealth] (-1,.5) -- (-1,-.5);
\draw[-,thick,>=stealth] (1,.5) -- (1,-.5);
\end{tikzpicture}
\end{subfigure}
\begin{subfigure}[b]{.3\textwidth}
\centering
\begin{tikzpicture}
\draw[<->,thick,>=stealth] (-1,0) -- node[invisible,label=above:{\tiny $P_a$},midway] {} (1,0);
\draw[->,thick,>=stealth,blue] (-1,-.5) -- (-1,0);
\draw[-,thick,>=stealth,blue] (1,.5) -- (1,0) -- (1.5,0);
\draw[->,thick,>=stealth,red] (1,-.5) -- (1,0);
\draw[-,thick,>=stealth,red] (-1,.5) -- (-1,0) -- (-1.5,0);
\end{tikzpicture}
\end{subfigure}
\begin{subfigure}[b]{.3\textwidth}
\centering
\begin{tikzpicture}
\draw[<->,thick,>=stealth] (-1,0) -- node[invisible,label=above:{\tiny $P_a$},midway] {} (1,0);
\draw[->,thick,>=stealth,blue] (-1.5,0) -- (-1,0);
\draw[-,thick,>=stealth] (-1,.5) -- (-1,-.5);
\draw[-,thick,>=stealth,blue] (1,.5) -- (1,0) -- (1.5,0);
\draw[->,thick,>=stealth] (1,-.5) -- (1,0);
\end{tikzpicture}
\end{subfigure}
\caption{Examples where two paths cannot touch when $P_a$ has no bend.}
\label{samesub}
\end{figure}
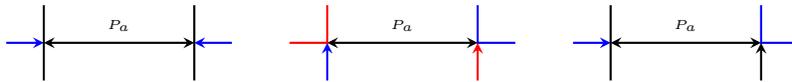

If $P_a$ has a bend, assume without loss of generality that the horizontal segment of $P_a$ contains $p_1$ and lies to the left and below of its vertical segment (see Fig. \ref{fc2}). If $p_1$ is a grid-point of type II.b, then $p_1$ is the bend-point or an endpoint of a path using the grid-edge below $p_1$. But then, since $p_2$ is assumed to be grid-point of type II, it is similarly an endpoint or the bend-point of a path using the grid-edge above $p_2$; and we conclude by Observation \ref{obs1} that those two paths cannot touch. We conclude by symmetry that $p_2$ can neither be a grid-point of type II.b and assume henceforth that both $p_1$ and $p_2$ are grid-points of type II.a. Now since $d(a) = 5$, there exists a path $P_b$ touching $P_a$ in its interior. If $P_b$ touches the vertical (resp. horizontal) segment of $P_a$ and lies to its right (resp. below) then, by Observation \ref{obs1}, $P_b$ cannot touch the other path having $p_1$ (resp. $p_2$) as an endpoint (see Fig. \ref{fc21} where colors are used to specify which paths cannot touch); and, if $P_b$ touches the vertical segment of $P_a$ and lies to its left, or if $P_b$ touches the horizontal segment of $P_a$ and lies above it, then by Observation \ref{obs2}, $P_b$ and the other paths having respectively $p_1$ and $p_2$ as an endpoint can then not pairwise touch (see Fig. \ref{fc22} where colors are used to specify which paths can not pairwise touch), which concludes the first part of this proof.

\begin{figure}[ht]
\centering
\begin{subfigure}[b]{.45\textwidth}
\centering
\begin{tikzpicture}[scale=0.7]
\node[invisible, label=below right:{\tiny $P_a$}] (Pa) at (0,0) {};

\draw[<-,thick,>=stealth] (-2,0)--(0,0);
\draw[->,thick,>=stealth] (0,0)--(0,2);
\draw[-,thick,>=stealth] (-2,1)--(-2,-1);
\draw[-,thick,>=stealth] (-1,2)--(1,2);
\draw[->,thick,>=stealth,blue] (-3,0)--(-2,0);
\draw[<-,thick,>=stealth] (0,2)--(0,3);
\draw[<-,thick,>=stealth,blue] (0,1)--(1,1) node[right] {\tiny $P_b$};
\end{tikzpicture}
\caption{$P_b$ touches $P_a$ on the right.}
\label{fc21}
\end{subfigure}
\hspace*{1cm}
\begin{subfigure}[b]{.45\textwidth}
\centering
\begin{tikzpicture}[scale=0.7]
\node[invisible,label=below right:{\tiny $P_a$}] (Pa) at (0,0) {};

\draw[<-,thick,>=stealth] (-2,0)--(0,0);
\draw[->,thick,>=stealth] (0,0)--(0,2);
\draw[-,thick,>=stealth] (-2,1)--(-2,-1);
\draw[-,thick,>=stealth] (-1,2)--(1,2);
\draw[->,thick,>=stealth,blue] (-3,0)--(-2,0);
\draw[<-,thick,>=stealth,blue] (0,2)--(0,3);
\draw[<-,thick,>=stealth,blue] (0,1)--(-1,1) node[left] {\tiny $P_b$};
\end{tikzpicture}
\caption{$P_b$ touches $P_a$ on the left.}
\label{fc22}
\end{subfigure}
\caption{Examples where two paths cannot touch when $P_a$ has a bend.}
\label{fc2}
\end{figure}
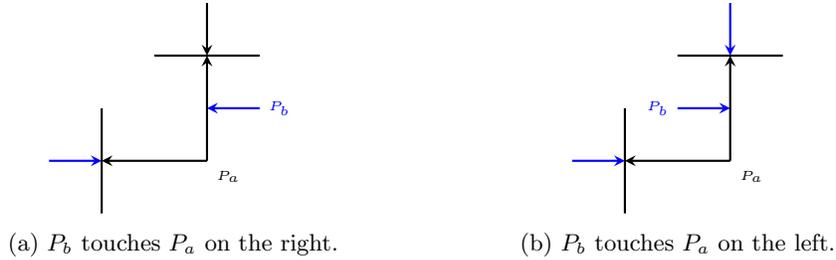

Now suppose that $\mathcal{R}$ contains a grid-point $p$ of type II and assume without loss of generality that $p = p_1$. Since $p_2$ is not a grid-point of type I nor a grid-point of type II, there must exist a path $P$ touching $P_a$ in its interior and a path $P'$ either also touching $P_a$ in its interior or having $p_2$ as an endpoint. 

Let us first assume that $p_1$ is of type II.b. Then, if $P_a$ has no bend, we conclude by Observation \ref{obs1} that $P$ cannot touch one of the paths touching $P_a$ at $p_1$ as the latter is the bend-point and an endpoint of two distinct paths using grid-edges orthogonal to $P_a$. Hence, $P_a$ must have a bend and as previously, we may assume without loss of generality that the horizontal segment of $P_a$ contains $p_1$ and lies to the left and below of its vertical segment. By Observation \ref{obs1}, we know that $P$ cannot touch the vertical segment of $P_a$ from the right nor can it touch the horizontal segment of $P_a$ (see Fig. \ref{fc3}); and since the same holds for $P'$, we conclude by Observation \ref{obs2} that $P$, $P'$ and the path with endpoint $p_1$ using the grid-edge below $p_1$ cannot pairwise touch (see Fig. \ref{fc32} where colors are used to specify which paths cannot pairwise touch).

\begin{figure}[ht]
\centering
\begin{subfigure}[b]{.45\textwidth}
\centering
\begin{tikzpicture}[scale=.7]
\node[invisible,label=below right:{\tiny $P_a$}] (Pa) at (0,0) {};

\draw[<-,thick,>=stealth] (-2,0)--(0,0);
\draw[-,thick,>=stealth] (0,0)--(0,2);
\draw[-,thick,>=stealth] (-2,1)--(-2,-1);
\draw[->,thick,>=stealth] (-2,-1)--(-2,0);
\draw[-,thick,>=stealth] (-3,0)--(-2,0)--(-2,1);
\draw[<-,thick,>=stealth] (0,1)--(-1,1) node[left] {\tiny $P$};
\end{tikzpicture}
\caption{The unique possibility for $P$ to touch $P_a$.}
\label{fc3}
\end{subfigure}
\hspace*{1cm}
\begin{subfigure}[b]{.45\textwidth}
\centering
\begin{tikzpicture}[scale=.7]
\node[invisible,label=below right:{\tiny $P_a$}] (Pa) at (0,0) {};

\draw[<-,thick,>=stealth] (-2,0)--(0,0);
\draw[->,thick,>=stealth] (0,0)--(0,2);
\draw[-,thick,>=stealth] (-2,1)--(-2,-1);
\draw[->,thick,>=stealth,blue] (-2,-1)--(-2,0);
\draw[-,thick,>=stealth] (-3,0)--(-2,0)--(-2,1);
\draw[<-,thick,>=stealth,blue] (0,1)--(-1,1) node[left] {\tiny $P$};
\draw[<-,thick,>=stealth,blue] (0,2)--(-1,2) node[left] {\tiny $P'$};
\end{tikzpicture}
\caption{Example where three paths cannot pairwise touch.}
\label{fc32}
\end{subfigure}
\caption{$P_a$ has a bend and $p_1$ is of type II.b.} 
\end{figure}
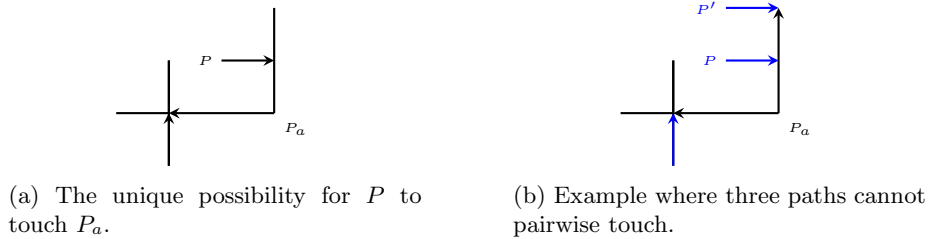

Assume henceforth that $p_1$ is of type II.a and denote by $P_b$ the other path having $p_1$ as an endpoint. If $P_a$ has no bend, by Observation \ref{obs1}, $P'$ must then touch $P_a$ orthogonally as it would otherwise not be able to touch $P_b$. If $P$ lies on the opposite side of $P_a$ than $P'$, it must then have by Observation \ref{obs1}, a common endpoint $p$ with $P'$ belonging to the interior of $P_a$ i.e. $p$ is a grid-point of type II.a; but then, it is clear that $P_b$ can not touch both $P$ and $P'$ (see Fig. \ref{fc8}). 

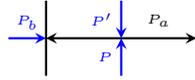
\begin{figure}[ht]
\centering
\begin{tikzpicture}
\draw[<->,thick,>=stealth] (-1,0) -- node[invisible,label=above:{\tiny $P_a$},near end] {} (1,0);
\draw[->,thick,>=stealth,blue] (-1.5,0) -- (-1,0) node[midway,above] {\tiny $P_b$};
\draw[-,thick,>=stealth] (-1,.5) -- (-1,-.5);
\draw[->,thick,>=stealth,blue] (0,-.5) -- (0,0) node[midway,left] {\tiny $P$};
\draw[->,thick,>=stealth,blue] (0,.5) -- (0,0) node[midway,left] {\tiny $P'$};
\end{tikzpicture}
\caption{$P_b$ cannot touch both $P$ and $P'$.}
\label{fc8}
\end{figure}

If now $P$ lies on the same side of $P_a$ as $P'$, we conclude by Observation \ref{obs2} that $P$, $P'$ and $P_b$ cannot pairwise touch (see Fig. \ref{fc4} where colors are used to specify which paths cannot pairwise touch).

\begin{figure}[ht]
\centering
\begin{tikzpicture}
\draw[<-,thick,>=stealth] (-1,0) -- node[invisible,label=above:{\tiny $P_a$},midway] {} (1,0);
\draw[->,thick,>=stealth,blue] (-1.5,0) -- (-1,0) node[midway,below] {\tiny $P_b$};
\draw[-,thick,>=stealth] (-1,.5) -- (-1,-.5);
\draw[->,thick,>=stealth,blue] (0,-.5) -- (0,0) node[midway,right] {\tiny $P$};
\draw[->,thick,>=stealth,blue] (1,-.5) -- (1,0) node[midway,right] {\tiny $P'$};
\end{tikzpicture}
\caption{$P$ lies on the same side of $P_a$ as $P'$.}
\label{fc4}
\end{figure}
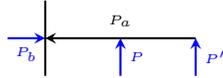

Hence, $P_a$  must have a bend and as previously, we may assume without loss of generality that the horizontal segment of $P_a$ contains $p_1$ and lies to the left  and below of its vertical segment. First assume that $P'$ has $p_2$ as an endpoint. Note that by Observation \ref{obs1}, $P'$ and $P$ cannot touch the vertical segment of $P_a$ from the right. If $P'$ uses the same column as $P_a$, then by Observation \ref{obs1} $P$ cannot touch the horizontal segment of $P_a$ from below; and we conclude by Observation \ref{obs2} that $P$, $P'$ and $P_b$ can then not pairwise touch (see Fig. \ref{fc6} where colors are used to specify which paths cannot pairwise touch).

\begin{figure}
\centering
\begin{subfigure}[b]{.45\textwidth}
\centering
\begin{tikzpicture}[scale=.7]s
\node[invisible,label=below right:{\tiny $P_a$}] (Pa) at (0,0) {};

\draw[<-,thick,>=stealth] (-2,0)--(0,0);
\draw[->,thick,>=stealth] (0,0)--(0,2);
\draw[-,thick,>=stealth] (-2,1)--(-2,-1);
\draw[-,thick,>=stealth] (-2,-1)--(-2,1);
\draw[->,thick,>=stealth,blue] (-3,0)--(-2,0);
\draw[<-,thick,>=stealth,blue] (0,2)--(0,3) node[midway,left] {\tiny $P'$};
\draw[->,thick,>=stealth,blue] (-1,1)--(0,1) node[midway,above] {\tiny $P$};
\end{tikzpicture}
\end{subfigure}
\hspace*{1cm}
\begin{subfigure}[b]{.45\textwidth}
\centering
\begin{tikzpicture}[scale=.7]
\node[invisible,label=below right:{\tiny $P_a$}] (Pa) at (0,0) {};

\draw[<-,thick,>=stealth] (-2,0)--(0,0);
\draw[->,thick,>=stealth] (0,0)--(0,2);
\draw[-,thick,>=stealth] (-2,1)--(-2,-1);
\draw[-,thick,>=stealth] (-2,-1)--(-2,1);
\draw[->,thick,>=stealth,blue] (-3,0)--(-2,0);
\draw[<-,thick,>=stealth,blue] (0,2)--(0,3) node[midway,left] {\tiny $P'$};
\draw[->,thick,>=stealth,blue] (-1,1)--(-1,0) node[midway,left] {\tiny $P$};
\end{tikzpicture}
\end{subfigure}
\caption{$P'$ lies on the same column as $P_a$.}
\label{fc6}
\end{figure}
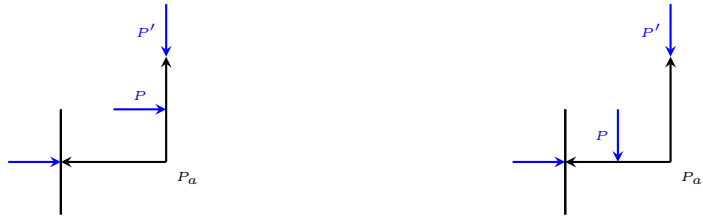

Hence, if $P''$ denotes the path not yet considered, $P''$ must touch the interior of $P_a$. By Observation \ref{obs1}, $P''$ cannot touch the vertical segment of $P_a$ from the right; and if both $P$ and $P''$ touch $P_a$ at a grid-point $p$ belonging to the horizontal segment of $P_a$, i.e. $p$ is a grid-point of type II.a, then it is clear that $P_b$ can not touch both $P$ and $P''$. Consequently, if $\mathcal{R}'$ denotes the representation obtained from $\mathcal{R}$ by deleting the path having $p_1$ as an interior point, then $\mathcal{R}'$ represents $K_5$ and we have by Observations \ref{wui} and \ref{lowerbound}:
\[10 = |E(K_5)| \leq w_P^{\mathcal{R}'} + w_{P''}^{\mathcal{R}'} + w_{P'}^{\mathcal{R}'} + w_{P_a}^{\mathcal{R}'} + w_{P_b}^{\mathcal{R}'} \leq 2(1+\frac{3}{2}) + (\frac{1}{2} + \frac{3}{2}) + 1 + (\frac{1}{2} + 1) = 9.5\]
where $w_Q^{\mathcal{R}'}$ denote the weight with respect to $\mathcal{R}'$ of the vertex whose corresponding path is $Q$. (Note that since $p_1$ is a grid-point of type II, the other endpoint of $P_b$ cannot be a grid-point of type II and recall that $p_2$ is an endpoint of $P'$). Hence, $P'$ does not have $p_2$ as an endpoint and must touch $P_a$ in its interior. But then, by reason of the foregoing, $p_2$ must be an interior point of $P''$. Now, if $P$ touches the horizontal segment of $P_a$ at a grid-point $p$, first assume it is from below. Then $p$ cannot be an endpoint of $P'$ ($P_b$ would otherwise not be able to touch both $P$ and $P'$) and $P'$ can neither touch the horizontal segment of $P_a$ from above by Observation \ref{obs1}, nor can it touch the horizontal segment of $P_a$ from below by Observation \ref{obs2} ($P$, $P'$ and $P_b$ would otherwise not be able to pairwise touch). Hence, $P'$ must touch the vertical segment of $P_a$ (from the left by Observation \ref{obs1}) and we conclude by Observation \ref{obs2} that $P$, $P'$ and $P_b$ can then not pairwise touch. Consequently, $P$ must touch the horizontal segment of $P_a$ from above and as previously, we conclude that $P'$ can then not touch the horizontal segment of $P_a$. Thus, the only possible configuration when $P$ touches the horizontal segment of $P_a$ is as shown in Fig. \ref{fc91} (the role of $P$ and $P'$ may be inverted). Now, if $P$ touches the vertical segment of $P_a$ (it must then be from the left), we may assume, by the foregoing case, that $P'$ also touches the vertical segment of $P_a$ (see Fig. \ref{fc92}).

\begin{figure}[ht]
\centering
\begin{subfigure}[b]{.45\textwidth}
\centering
\begin{tikzpicture}[scale=.7]
\node[invisible,label=below right:{\tiny $P_a$}] (Pa) at (0,0) {};

\draw[<-,thick,>=stealth] (-2,0)--(0,0);
\draw[->,thick,>=stealth] (0,0)--(0,2);
\draw[-,thick,>=stealth] (-2,1)--(-2,-1);
\draw[-,thick,>=stealth] (-2,-1)--(-2,1);
\draw[->,thick,>=stealth] (-3,0)--(-2,0) node[midway,below] {\tiny $P_b$};
\draw[->,thick,>=stealth] (-.8,1)--(0,1);
\draw[->,thick,>=stealth] (-1,.8)--(-1,0);
\draw[-,thick,>=stealth] (-1,2)--(1,2) node[midway,above] {\tiny $P''$};
\end{tikzpicture}
\caption{The unique possibility when a path touches the horizontal segment of $P_a$.}
\label{fc91}
\end{subfigure}
\hspace*{1cm}
\begin{subfigure}[b]{.45\textwidth}
\centering
\begin{tikzpicture}[scale=.7]
\node[invisible,label=below right:{\tiny $P_a$}] (Pa) at (0,0) {};

\draw[<-,thick,>=stealth] (-2,0)--(0,0);
\draw[->,thick,>=stealth] (0,0)--(0,2);
\draw[-,thick,>=stealth] (-2,1)--(-2,-1);
\draw[-,thick,>=stealth] (-2,-1)--(-2,1);
\draw[->,thick,>=stealth] (-3,0)--(-2,0) node[midway,below] {\tiny $P_b$};
\draw[->,thick,>=stealth] (-1,1.25)--(0,1.25) node[midway,above] {\tiny $P'$};
\draw[->,thick,>=stealth] (-1,.75)--(0,.75) node[midway,below] {\tiny $P$};
\draw[-,thick,>=stealth] (-1,2)--(1,2) node[midway,above] {\tiny $P''$};
\end{tikzpicture}
\caption{The unique possibility when no path touches the horizontal segment of $P_a$.}
\label{fc92}
\end{subfigure}
\caption{The only possible configurations when $p_2$ is an interior point of $P''$.}
\label{fc9}
\end{figure}
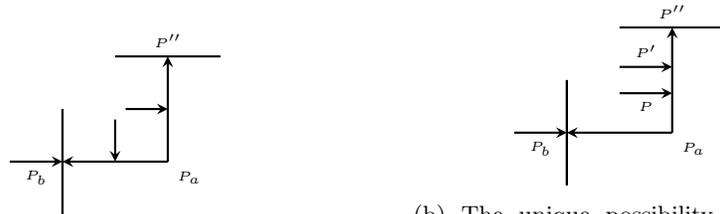

In both cases (see Fig. \ref{fc9}), $P''$ has an endpoint lying to the right of the vertical segment of $P_a$ which can only belong to the path having $p_1$ as an interior point. Hence, if $\mathcal{R}'$ denotes the representation obtained from $\mathcal{R}$ by deleting this path, then $\mathcal{R}'$ represents $K_5$ and we have by Observations \ref{wui} and \ref{lowerbound} (note that since $p_1$ is a grid-point of type II, the other endpoint of $P_b$ cannot be a grid-point of type II):
\[10 = |E(K_5)| \leq w_P^{\mathcal{R}'} + w_{P'}^{\mathcal{R}'} + w_{P''}^{\mathcal{R}'} + w_{P_a}^{\mathcal{R}'} + w_{P_b}^{\mathcal{R}'} \leq 2(1+\frac{3}{2}) + \frac{3}{2} + 2(\frac{1}{2} + 1) = 9.5\]
a contradiction which concludes the proof of Claim 2. 
\end{cproof}\\

It now follows from Observation \ref{wui} that the weight of every vertex with respect to $\mathcal{R}$ is at most 2. But then, we have by Observation \ref{lowerbound}:
\[ 15 = |E| \leq \sum\limits_{u \in V} w_u \leq 2|V| = 12\]
Hence, $K_6$ cannot be $B_1$-CPG. \qed 

\subsection{An illustration of the construction in the proof of Theorem \ref{thm:3c}}
\label{app:ex}

In Fig. \ref{fullexample} is given an example of how the representation constructed in the proof may be locally. Note that edges are dealt with arbitrarily and thus $u$ may correspond either to a $u_1$, as for instance for edge $uw$, or a $u_k$ with $k >1$, as for instance for edge $uv$.

\begin{figure}[htb]
\centering
\begin{subfigure}[b]{.2\textwidth}
\centering
\begin{tikzpicture}[scale=.1]
\node[invisible] at (0,-14) {};
\node[circ,label=left:{$v$}] at (0,0) {};
\node[circ,label=above right:{$u$}] at (12,5) {};
\node[circ] at (20,5) {};
\node[circ,label=right:{$w$}] at (21,17) {};
\node[circ] at (12,-2) {};
\draw[-] (0,0) -- (6,0) -- (6,5) -- (20,5);
\draw[-] (12,-2) -- (12,10) -- (17,10) -- (17,14) -- (21,14) -- (21,17);
\end{tikzpicture}
\caption{Local embedding of $G$.}
\end{subfigure}
\hspace*{.5cm}
\begin{subfigure}[b]{.7\textwidth}
\centering
\begin{tikzpicture}[scale=.4]
\node[circ,label=above right:{$u$}] at (12,5) {};
\node[circ] at (20,5) {};
\node[circ,label=right:{$w$}] at (21,17) {};
\node[circ] at (12,-2) {};
\node[circ,label=left:{$v$}] at (0,0) {};
\draw[thick,dashed,red] (0,1.5) -- (4.5,1.5) -- (4.5,6.5) -- (10.5,6.5) -- (10.5,11.5) -- (15.5,11.5) -- (15.5,15.5) -- (19.5,15.5) -- (19.5,16);
\draw[thick,dashed,red] (0,-1.5) -- (7.5,-1.5) -- (7.5,3.5) -- (10.5,3.5) -- (10.5,-1);
\draw[thick,dashed,red] (13.5,-1) -- (13.5,3.5) -- (20,3.5);
\draw[thick,dashed,red] (20,6.5) -- (13.5,6.5) -- (13.5,8.5) -- (18.5,8.5) -- (18.5,12.5) -- (22.5,12.5) -- (22.5,16);
\draw[-,thick,blue] (0,-1.7) -- (0,1.7);
\draw[thick,dotted,blue] (0,-2.2) -- (0,2.2);
\draw[<->,thick,>=stealth,blue] (6,-1.5) -- (6,6.5);
\draw[<->,thick,>=stealth,blue] (12,1.5) -- (12,11.5);
\draw[->,thick,>=stealth,blue] (12,-2) -- (12,1.5);
\draw[thick,dotted,blue] (12,-2) -- (12,-2.5);
\draw[-,thick,blue] (20,3.3) -- (20,6.7);
\draw[thick,dashed,blue] (20,2.8) -- (20,7.2);
\draw[<->,thick,>=stealth,blue] (17,8.5) -- (17,15.5);
\draw[<-,thick,>=stealth,blue] (21,12.5) -- (21,17.5);
\draw[thick,dotted,blue] (21,17.5) -- (21,18);

\draw[<->,thick,>=stealth] (0,-.2) -- (3,-.2);
\draw[<->,thick,>=stealth] (3,-.2) -- (6,-.2);
\draw[<->,thick,>=stealth] (0,-.7) -- (6,-.7);
\draw[<->,thick,>=stealth] (0,-1.2) -- (3,-1.2);
\draw[<->,thick,>=stealth] (3,-1.2) -- (6,-1.2);
\draw[<->,thick,>=stealth] (3,-.2) -- (3,-.7);
\draw[<->,thick,>=stealth] (3,-.7) -- (3,-1.2);

\draw[<->,thick,>=stealth] (4.8,5.2) -- (4.8,6.2);
\draw[<->,thick,>=stealth] (5.4,5.2) -- (5.4,5.7);
\draw[<->,thick,>=stealth] (5.4,5.7) -- (5.4,6.2);
\draw[<->,thick,>=stealth] (4.8,5.2) -- (5.4,5.2);
\draw[<->,thick,>=stealth] (5.4,5.2) -- (6,5.2);
\draw[<->,thick,>=stealth] (4.8,5.7) -- (6,5.7);
\draw[<->,thick,>=stealth] (4.8,6.2) -- (5.4,6.2);
\draw[<->,thick,>=stealth] (5.4,6.2) -- (6,6.2);
\draw[<->,thick,>=stealth] (6,5.2) -- (12,5.2);
\draw[<->,thick,>=stealth] (6,5.7) -- (12,5.7);
\draw[<->,thick,>=stealth] (6,6.2) -- (12,6.2);

\draw[<->,thick,>=stealth] (12,3.8) -- (16,3.8);
\draw[<->,thick,>=stealth] (16,3.8) -- (20,3.8);
\draw[<->,thick,>=stealth] (12,4.3) -- (20,4.3);
\draw[<->,thick,>=stealth] (12,4.8) -- (16,4.8);
\draw[<->,thick,>=stealth] (16,4.8) -- (20,4.8);
\draw[<->,thick,>=stealth] (16,3.8) -- (16,4.3);
\draw[<->,thick,>=stealth] (16,4.3) -- (16,4.8);

\draw[<->,thick,>=stealth] (12,10.2) -- (14.5,10.2);
\draw[<->,thick,>=stealth] (14.5,10.2) -- (17,10.2);
\draw[<->,thick,>=stealth] (12,10.7) -- (17,10.7);
\draw[<->,thick,>=stealth] (12,11.2) -- (14.5,11.2);
\draw[<->,thick,>=stealth] (14.5,11.2) -- (17,11.2);
\draw[<->,thick,>=stealth] (14.5,10.2) -- (14.5,10.7);
\draw[<->,thick,>=stealth] (14.5,10.7) -- (14.5,11.2);

\draw[<->,thick,>=stealth] (15.8,14.2) -- (15.8,15.2);
\draw[<->,thick,>=stealth] (16.4,14.2) -- (16.4,14.7);
\draw[<->,thick,>=stealth] (16.4,14.7) -- (16.4,15.2);
\draw[<->,thick,>=stealth] (15.8,14.2) -- (16.4,14.2);
\draw[<->,thick,>=stealth] (16.4,14.2) -- (17,14.2);
\draw[<->,thick,>=stealth] (15.8,14.7) -- (17,14.7);
\draw[<->,thick,>=stealth] (15.8,15.2) -- (16.4,15.2);
\draw[<->,thick,>=stealth] (16.4,15.2) -- (17,15.2);
\draw[<->,thick,>=stealth] (17,14.2) -- (21,14.2);
\draw[<->,thick,>=stealth] (17,14.7) -- (21,14.7);
\draw[<->,thick,>=stealth] (17,15.2) -- (21,15.2);
\end{tikzpicture}
\caption{The corresponding local representation of $G'$.}
\end{subfigure}
\caption{An example of the transformations for a vertex $u$ of $V$.}
\label{fullexample}
\end{figure}
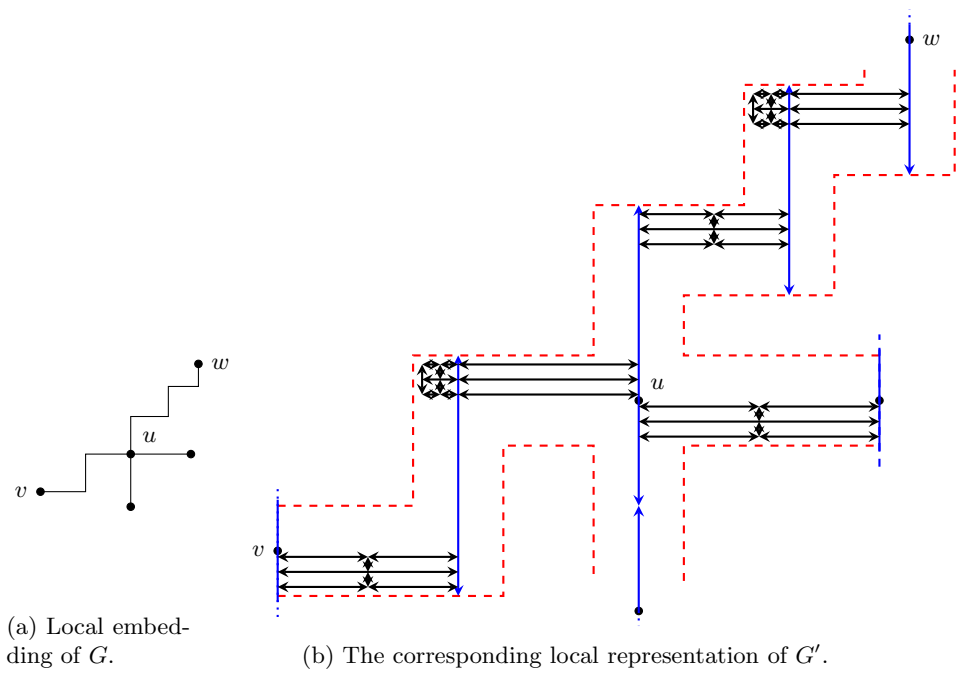

\end{document}